\documentclass[12pt,onecolumn,twoside]{IEEEtran}

\usepackage{calc}

\usepackage{multirow}
\usepackage{cite}
\usepackage{graphicx,subfigure}
\usepackage{psfrag}
\usepackage{amsmath,amssymb,amsthm}
\usepackage{color}
\usepackage{tikz} 

\usepackage{cases}

\DeclareMathOperator*{\defeq}{\triangleq}

\newtheorem{theorem}{Theorem}
\newtheorem{corollary}{Corollary}

\newtheorem{lemma}{Lemma}

\newtheorem{example}{Example}

\newcommand{\bit}{\begin{itemize}}
\newcommand{\eit}{\end{itemize}}

\newcommand{\bc}{\begin{center}}
\newcommand{\ec}{\end{center}}

\newcommand{\ba}{\begin{array}}
\newcommand{\ea}{\end{array}}

\newcommand{\beq}{\begin{equation}}
\newcommand{\eeq}{\end{equation}}

\newcommand{\beqn}{\begin{equation*}}
\newcommand{\eeqn}{\end{equation*}}

\newcommand{\bean}{\begin{eqnarray*}}
\newcommand{\eean}{\end{eqnarray*}}
\newcommand{\bea}{\begin{eqnarray}}
\newcommand{\eea}{\end{eqnarray}}

\def\E{\mathbb{E}}

\def\F{\mathbb{F}}

\def\av{\boldsymbol{a}}
\def\bv{\boldsymbol{b}}

\def\sv{\boldsymbol{s}}

\def\xv{\boldsymbol{x}}
\def\yv{\boldsymbol{y}}

\newcommand{\Cc}{{\mathcal C}}

\newcommand{\Nc}{{\mathcal N}}

\newcommand{\T}{{\scriptscriptstyle\mathsf{T}}}

\newtheorem{remark}{Remark}

\renewcommand{\Bmatrix}[1]{\begin{bmatrix}#1\end{bmatrix}}

\newcommand{\non}{\nonumber}

\newcommand{\Hen}{\mathbb{H}}
\newcommand{\hen}{\mathrm{h}}
\newcommand{\Imu}{\mathbb{I}}

\newcommand{\bln}{n}

\newcommand{\mc}{m_c}
\newcommand{\md}{m_d}
\newcommand{\J}{\mathrm{J}}

\IEEEoverridecommandlockouts

\begin{document}
\sloppy

\title{ New Results on the Secure Capacity of Symmetric Two-User Interference Channels}
\title{ On the Optimality of Secure Communication Without Using Cooperative Jamming}
\author{Jinyuan Chen 
\thanks{Jinyuan Chen is with Louisiana Tech University, Department of Electrical Engineering, Ruston, USA (email: jinyuan@latech.edu). This work was presented in part at 54th Annual Allerton Conference on Communication, Control, and Computing, 2016.} 
}

\maketitle
\pagestyle{headings}

\begin{abstract}

We consider secure communication over a two-user Gaussian interference channel, where each transmitter sends a \emph{confidential} message  to its legitimate receiver.  
For this setting,  we identify a regime where  the simple scheme of using Gaussian wiretap codebook at each transmitter (without cooperative jamming) and treating interference as noise at each intended receiver (in short, GWC-TIN scheme) achieves the optimal secure sum capacity to within a constant gap. 
The results are proved by first considering the deterministic interference channel model and identifying a regime in which  a simple scheme without using cooperative jamming is optimal in terms of secure sum capacity.
For the symmetric case of the deterministic model, this simple scheme is optimal \emph{if and only if} the interference-to-signal ratio (in channel strengths) is no more than $2/3$. 

\end{abstract}

\section{Introduction}

The notion of information-theoretic secrecy was first introduced by Shannon  in his seminal work \cite{Shannon:49}, which studied a secure communication in the presence of a private key that is revealed to both transmitter and legitimate receiver  but not to the eavesdropper. 
Later,  Wyner introduced the notion of secure capacity via  a degraded  wiretap channel, in which a  transmitter intends to send a confidential message to a legitimate receiver by hiding it from a degraded eavesdropper \cite{Wyner:75}. 
The secure capacity is the maximum rate at which the confidential message can be transmitted reliably and securely to the legitimate receiver. 
Wyner's result  was subsequently generalized to the non-degraded wiretap channel by Csisz\`ar and  K{\"o}rner \cite{CsiszarKorner:78}, and the Gaussian wiretap channel by Leung-Yan-Cheong and Hellman \cite{CH:78}.
This line of  secure capacity research has been extended to many multiuser channels, most notably, the broadcast  channels \cite{LMSY:08, LLPS:10, XCC:09,ChiaGamal:12,KTW:08}, multiple access channels  \cite{TY:08cj, TekinYener:08d, LP:08, LLP:11, KG:15, HKY:13},  and the interference channels \cite{LMSY:08,LBPSV:09,LYT:08,HY:09,YTL:08, KGLP:11, XU:14, XU:15, MDHS:14, MM:14o, GTJ:15,MXU:17}.

In the line of secure capacity research, cooperative jamming has been proposed extensively to improve the achievable secure rates in many channels (see \cite{TY:08cj,LMSY:08, XU:14, XU:15} and references therein). 
In particular,  cooperative jamming has been proposed in \cite{XU:14} and \cite{XU:15} to achieve the optimal secure sum degrees-of-freedom (DoF) in the interference channel with confidential messages,  wiretap channel with helpers, multiple access wiretap channel, and the broadcast channel with confidential messages. 
The basic idea of the cooperative jamming scheme is to send jamming signals to confuse the potential eavesdroppers, while keeping legitimate receivers' abilities to decode the desired messages. This might involve a cooperation between the transmitters, and a careful design on the \emph{direction} and/or \emph{power}  of the cooperative jamming signals (see \cite{TY:08cj, LMSY:08, XU:14, XU:15}).  It is therefore implicit that the cooperative jamming schemes might  incur  some extra overhead, e.g., due to network coordination,  channel state information (CSI) acquisition, and power consumption.

In this work we study the secure communication schemes without cooperative jamming. In particular, for a two-user Gaussian interference channel with confidential messages, we identify a regime in which  the simple scheme of using Gaussian wiretap codebook at each transmitter, without cooperative jamming, and treating interference as noise at each intended receiver (in short, GWC-TIN scheme) achieves the optimal secure sum capacity to within a constant gap. 
The secrecy offered by this GWC-TIN scheme is information-theoretic secrecy, which holds for any decoding method at any unintended receiver (eavesdropper). 
In this simple scheme, the transmitters do \emph{not} need to know the information of  the \emph{channel phases}. Therefore, the overhead associated with acquiring  channel state information at the transmitters (CSIT) is minimal for the  GWC-TIN scheme.

The results are proved by first considering the deterministic interference channel model (see~\cite{ADT:11}) and identifying a regime in which  a simple scheme without using cooperative jamming is optimal in terms of secure sum capacity.
In this simple scheme, the data is simply transmitted over the least significant signal bits such that no interference is leaked to the unintended receiver.  In this way, the deterministic interference channel is decomposed  into  two parallel channels ---  in each  channel the transmitter sends confidential data to its legitimate receiver without the cooperation from the other transmitter.
For the symmetric case of the deterministic model, this simple scheme is optimal \emph{if and only if} the interference-to-signal ratio (in channel strengths) is no more than $2/3$. 

To prove the optimality of the aforementioned schemes, we derive a new secure capacity bound for each of the two interference channel models.   In our proof the approach is different from the genie-aided approach that is commonly used in the settings without secrecy constraints (see~\cite{ETW:08}). In the genie-aided approach,  some genie-aided information is typically provided to the receivers, which might give a loose bound in our setting.

The remainder of this work is organized as follows.  
Section~\ref{sec:system} describes the  system model, as well as the simple scheme without cooperative jamming, for each of  the Gaussian and deterministic interference channels.  
Section~\ref{sec:mainresult} provides the  main results of this work. 
The  proofs are provided in   Section~\ref{sec:converseDet}, Section~\ref{sec:converse} and the appendices.
The work is concluded in Section~\ref{sec:conclusion}.
Throughout this work, $\Imu(\bullet)$, $\Hen(\bullet)$ and $\hen(\bullet)$ denote the mutual information, entropy and differential entropy,  respectively.  
$(\bullet)^\T$ denotes the transpose operation.  
$\F^{q}_{2}$ denotes a set of $q$-tuples of binary numbers.  
$(\bullet)^+= \max\{0, \bullet\}$. Logarithms are in base~$2$.  
Unless for some specific parameters,  matrix,  scalar, and  vector are usually denoted by the  italic uppercase symbol (e.g., $S$),  italic lowercase symbol (e.g., $s$), and  the bold italic lowercase symbol  (e.g., $\sv$), respectively.
$s \sim \mathcal{CN}(0, \sigma^2)$  denotes that the random variable $s$ has a circularly symmetric complex normal distribution with zero mean and $\sigma^2$ variance.

\section{System models and preliminaries  \label{sec:system} }

This section provides the system models for Gaussian interference channel and deterministic interference channel, respectively.
For each model, a simple scheme without using cooperative jamming is also discussed in this section.

\subsection{Gaussian interference channel  \label{sec:systemGaussian} }

We begin with a two-user  Gaussian interference channel. The channel output at receiver~$k$ at time~$t$ is  
\begin{align}
y_{k}(t) &= \sum_{\ell=1}^{2}  \sqrt{P^{\alpha_{k\ell}}} e^{j\theta_{k\ell}} x_{\ell}(t) +z_{k}(t), \quad k=1,2,  \label{eq:channelG} 
\end{align}
$t=1,2, \cdots, \bln $, where  $x_{\ell}(t)$ is the channel input at transmitter~$\ell$  subject to a normalized power constraint  $\E |x_{\ell}(t)|^2 \leq 1$, $z_k(t) \sim \mathcal{C}\mathcal{N}(0, 1)$ is additive white Gaussian noise at receiver~$k$, $\sqrt{P^{\alpha_{k\ell}}}$ and $\theta_{k\ell}$  represent the magnitude and phase of the channel between transmitter~$\ell$ and receiver~$k$, where $P\geq 1$ is a nominal power value\footnote{In this work we assume that the channel phases, as well as the channel strengths, are fixed over the whole communication period. However, our results can be extended easily to the settings where the channel phases are time-varying.}. 
The  exponent  $\alpha_{k\ell} \geq 0$ represents the channel strength of the link between transmitter~$\ell$ and receiver~$k$.   
We assume that  each transmitter knows the channel strengths  $\{\alpha_{k\ell} \}_{k, \ell}$ but not necessarily the phases $\{\theta_{k\ell} \}_{k, \ell}$, while each receiver knows all the channel strengths and phases.

For this interference channel, each transmitter wishes to send a confidential message to its legitimate receiver. Specifically transmitter~$k$ wishes to send to  receiver~$k$ a message $w_k$ that is  uniformly chosen from a set $\mathcal{W}_k \defeq \{1,2,\cdots, 2^{\bln R_k}\}$, where $R_k$ is  the rate  (bits/channel use) of this message and $\bln$ is the total number of channel uses, $k=1,2$. 
At transmitter~$k$, a stochastic function \[f_k: \mathcal{W}_k  \to   \mathcal{X}_k^{\bln} , \quad  k=1,2 \] is employed to encode the message. 
A secure rate pair $(R_1, R_2)$ is said to be achievable  if for any $\epsilon >0$ there exists a sequence of $\bln$-length codes such that each receiver can decode its own message reliably, i.e., the probability of decoding error is less than $\epsilon$, 
 \begin{align}
 \text{Pr}[w_k  \neq \hat{w}_k  ]  \leq \epsilon, \quad \forall k   \label{eq:Pedef}
  \end{align}
and the messages are kept secret such that 
 \begin{align}
\frac{1}{\bln}\Hen(w_1 | y_{2}^{\bln})  &\geq   \frac{1}{\bln}\Hen(w_1)  - \epsilon  \label{eq:defsecrecy1}  \\
\frac{1}{\bln}\Hen(w_2 | y_{1}^{\bln})  &\geq   \frac{1}{\bln}\Hen(w_2)  - \epsilon,   \label{eq:defsecrecy2}
 \end{align}
where $y_{k}^{\bln}$ represents the $\bln$-length channel output  of receiver~$k$, $k=1,2$. 
The secure capacity region $C$ is the closure of the set of all achievable secure rate pairs.
The secure sum capacity is defined as: 
 \begin{align}
 C_{\text{sum}} \defeq \sup \big\{ R_1 + R_2 |  \  (R_1, R_2) \in C  \big\} .  \label{eq:defGDoFsum}
  \end{align}
The secure sum generalized degrees-of-freedom (GDoF) is defined as 
 \begin{align}
 d_{\text{sum}}  \defeq   \lim_{P \to \infty}   \frac{C_{\text{sum}}}{\log P}.  \label{eq:defGDoF}
 \end{align}

\subsection{Deterministic interference channel  \label{sec:systemDeterministic} }

One way to better understand the capacity of the Gaussian channels is to study their linear deterministic models (see \cite{ADT:11}). 
In this work we also consider a two-user  deterministic interference channel.  
By following the common convention (see~\cite{ADT:11, GTJ:15}),   we assume that  the  input-output relation of the deterministic channel is given by 
\begin{align}  
\yv_1(t) &=   S^{q- m_{11}} \xv_1(t)  \oplus   S^{q- m_{12}} \xv_2(t)         \label{eq:detTHIC1} \\
\yv_2(t) &=   S^{q- m_{21}} \xv_1(t)  \oplus   S^{q- m_{22}} \xv_2(t),           \label{eq:detTHIC2} 
\end{align}  
where $\xv_k(t) = \Bmatrix{ x_{k,1}(t) , x_{k,2}(t) , \cdots, x_{k,q}(t) }^\T \in
\F^{q}_{2}$ denotes the channel input of transmitter $k$ at time $t$;  $\yv_k(t)\in
\F^{q}_{2}$ denotes the channel output of receiver~$k$ at time $t$, $k=1,2$, $q \defeq \max\{m_{11}, m_{12}, m_{21}, m_{22}\}$;  $S$ is a $q\times q$ lower
shift matrix, and  $S^{q- m_{21}} \xv_1(t)= \Bmatrix{ 0,
  \cdots,0, x_{1,1}(t), \cdots, x_{1,m_{21}}(t) }^\T$. $\oplus$ denotes modulo~2 addition.  The nonnegative integers $m_{k k}$ and $m_{\ell k}$ denote the number of information bits that can be communicated per channel use over the direct and
cross links, respectively, for $\ell, k \in \{1,2\}, \ell \neq k$.  
For the \emph{symmetric} case of the deterministic channel model, we  let 
\begin{align}  
\md  =  m_{11}  =  m_{22} ,  \quad  \mc  =  m_{12}   =  m_{21}          \label{eq:detsym} 
\end{align} 
 and let  
\begin{align}  
\alpha \defeq \frac{\mc}{\md}         \label{eq:detsyma} 
\end{align} 
that is a normalized interference parameter.
 
 Similarly to the Gaussian case,  transmitter~$k$ wishes to send to its receiver~$k$ a message $w_k$ that is  uniformly chosen from a set $\mathcal{W}_k=\{1,2,\cdots, 2^{\bln R_k}\}$,  $k=1,2$. 
Transmitter~$k$ uses a stochastic function $g_k: \mathcal{W}_k \to  \F^{q \times \bln}_{2}$ to encode the message. 
A secure rate pair $(R_1, R_2)$ is said to be achievable  if for any $\epsilon >0$ there exists a sequence of $\bln$-length codes such that each receiver can decode its own message reliably (cf.~\eqref{eq:Pedef})
and the messages are kept secret, i.e.,  
$  \Imu(w_1; \yv_{2}^{\bln})  \leq  \bln \epsilon$ and $\Imu(w_2; \yv_{1}^{\bln})  \leq  \bln \epsilon$. 
The secure  capacity region $C$ and sum capacity $C_{\text{sum}}$ are defined  similarly as in the Gaussian case (cf.~\eqref{eq:defGDoFsum}).

\subsection{The scheme without using cooperative jamming  for the deterministic channel model  \label{sec:noCJdet} }

Let us discuss a simple scheme without using cooperative jamming (in short, WoCJ scheme) for the deterministic channel defined  in Section~\ref{sec:systemDeterministic}.  
In this scheme,  transmitter~$k$  simply sends a total of  $(m_{kk} - m_{\ell k})^+$ bits of private data over  the  least significant signal bits \emph{without using cooperative jamming}, for $k,\ell   \in \{1,2\}, k \neq \ell$ and $(\bullet)^+= \max\{0, \bullet\}$.  The transmission of the private data is secure from the unintended receiver (eavesdropper) because the private data is not seen by the unintended receiver. 
In this simple way the scheme achieves the following secure rate pair: 
\begin{align}
R_1  &  =  (m_{11 } - m_{21})^+       \label{eq:NoCJdetrate1} \\
R_2 &=    (m_{22 } - m_{12})^+.      \label{eq:NoCJdetrate2}
 \end{align}
Fig.~\ref{fig:det2131}  depicts the scheme  for a specific setting with $m_{11}=m_{22}= 3, m_{21}=2$, and $m_{12}=1$. 
Since the data is transmitted over the least significant signal bits, it implies that: 1)  no confidential information is leaked to the unintended receiver; 2)   no interference is leaked to the unintended receiver\footnote{When the  signal  is not intended to the receiver, it usually interferes  the desired signal and is typically called as interference. However, in some communication scenarios with  secrecy constraints, the interference signal could be utilized as a jamming signal to improve the secure rate of the system. For notational convenience we will still use the name of ``interference'' to denote the unintended signal.}.  In this way, the interference channel is decomposed  into  two parallel wiretap channels --- in each  wiretap channel the transmitter sends confidential data to its legitimate receiver without the cooperation from the other transmitter (see Fig.~\ref{fig:det2131}).

\begin{figure}[t!]
\centering
\includegraphics[width=8cm]{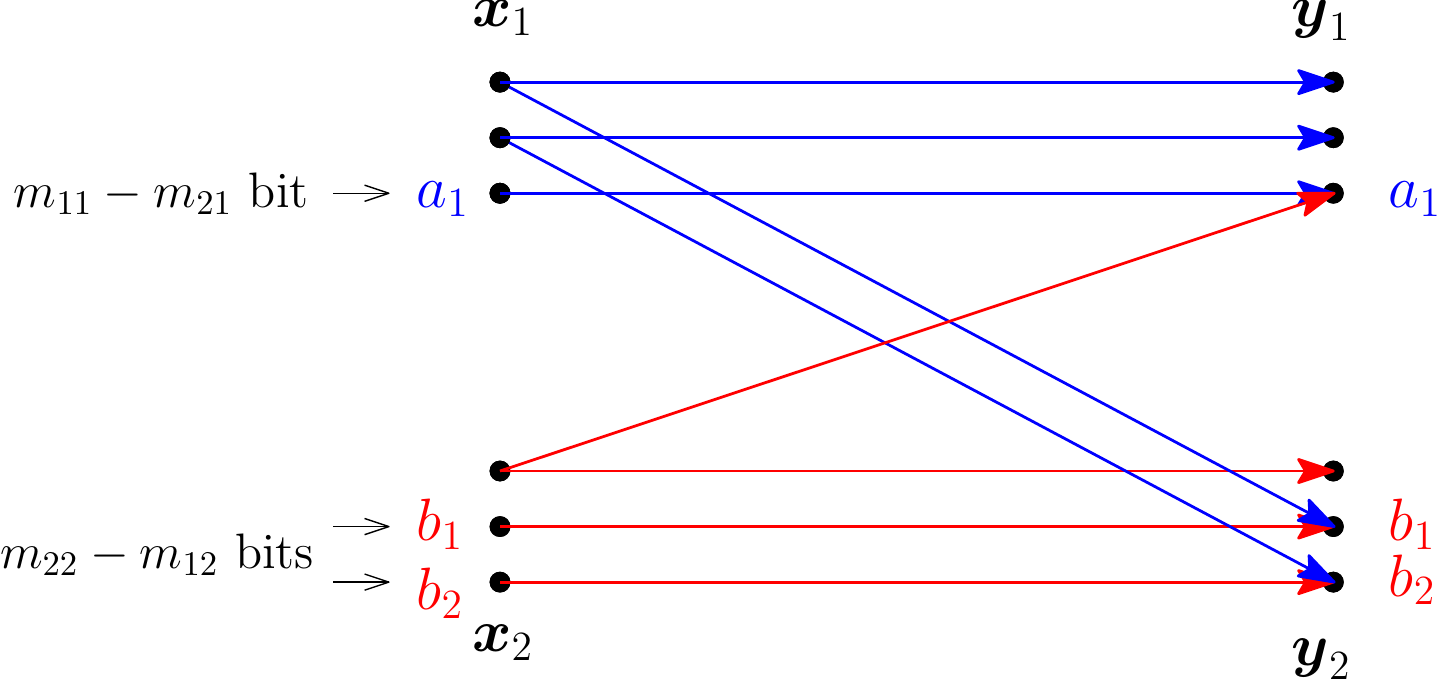}
\caption{The  scheme without using cooperative jamming for the deterministic channel: the case  with  $m_{11}=m_{22}= 3, m_{21}=2$, and $m_{12}=1$. In this case,  transmitter~1 simply sends its private data  over   $(m_{11 } - m_{21})^+=1$ least significant  bit of the signal, while  transmitter~2 sends its private data  over   $(m_{22 } - m_{12})^+=2$ least significant  bits of the signal, without using cooperative jamming. The transmission of private data is secure from the unintended receiver.}
 \label{fig:det2131}
\end{figure}

\subsection{The scheme without using cooperative jamming for the Gaussian channel model  \label{sec:noCJGau} }

The scheme discussed in Section~\ref{sec:noCJdet}  can be extended to the Gaussian channel in a  similar way.
For this Gaussian interference channel, each  interference signal leaked to the unintended receiver is scaled down to the noise level by applying a proper power allocation strategy.
Due to the noisy interference, the Gaussian interference channel is approximately decomposed  into  two parallel wiretap Gaussian channels. Therefore, in this scheme each transmitter simply employs a \emph{Gaussian wiretap codebook} (GWC) to guarantee the secrecy \emph{without using cooperative jamming}, while each receiver simply \emph{treats interference as noise} (TIN) when decoding its desired message. We call it as a  GWC-TIN scheme.  Note that the secrecy offered by this GWC-TIN scheme is information-theoretic secrecy, which holds for any decoding method at any eavesdropper. 
Some details of the scheme are discussed as follows.

\subsubsection{Gaussian wiretap codebook}
To build the  codebook, transmitter~$k$  generates a total of $2^{\bln (R_k + R_k')}$  independent  codewords  $v^{\bln}_k$  with each element independent and identically distributed (i.i.d.) according to a circularly-symmetric complex normal distribution with variance $P^{- \beta_k}$, $k=1,2$, for some $R_k, R_k'$ and $\beta_k \geq 0$ that will be designed specifically later on. 
The codebook $\mathcal{B}_{k}$ is defined as a set of the labeled codewords:
   \begin{align}
     \mathcal{B}_{k} \defeq \bigl\{ v^{\bln}_k (w_k,  w_k'): \  w_k \in \{1,2,\cdots, 2^{\bln R_k}\}, \  w_k' \in \{1,2,\cdots, 2^{\bln R_k'}\}   \bigr\},  \quad k=1,2.      \label{eq:code2341}
     \end{align}
 To transmit the message $w_k$,  transmitter~$k$ at first selects a bin (sub-codebook)  $\mathcal{B}_{k}( w_k) $ that is defined as 
\[   \mathcal{B}_{k} (w_k)  \defeq \bigl\{ v^{\bln}_k (w_k,  w_k'): \  w_k' \in \{1,2,\cdots, 2^{\bln R_k'}\}   \bigr\},  \quad k=1,2,  \]
and then \emph{randomly} chooses a codeword $v^{\bln}_k$ from the selected bin according to a uniform distribution.
Since this scheme will not use cooperative jamming, the chosen codeword $v^{\bln}_k$ will be mapped exactly as a channel input sequence by transmitter~$k$, that is,  $x_k (t) =  v_k (t), \ t=1,2, \cdots, \bln$,
where $v_k (t)$ is the   $t$th element of  the codeword $v^{\bln}_k$,  $k=1,2$. Based on this one-to-one mapping and Gaussian codebook, it implies that 
\begin{align}
x_k (t)= v_k (t) \sim \Cc\Nc (0, P^{- \beta_k}), \quad \forall t,  \quad   k=1,2.  \label{eq:map888}
\end{align}
Then,  the received signals take the following forms (removing the time index for simplicity):
\begin{align}
y_{1} &=  \underbrace{\sqrt{P^{\alpha_{11}}} e^{j\theta_{11}} v_1}_{P^{\alpha_{11}  - \beta_1}}+  \underbrace{\sqrt{P^{\alpha_{12}}} e^{j\theta_{12}}v_2}_{P^{\alpha_{12} - \beta_2}}  +  \underbrace{z_{1}}_{P^{0}}    \label{eq:y441}  \\
y_{2} &= \underbrace{\sqrt{P^{\alpha_{22}}} e^{j\theta_{22}} v_{2} }_{P^{\alpha_{22}-\beta_2}} + \underbrace{ \sqrt{P^{\alpha_{21}}} e^{j\theta_{21}} v_{1}}_{P^{\alpha_{21}-\beta_1}}    + \underbrace{z_{2}}_{P^{0}} \label{eq:y442}
\end{align}
(cf.~\eqref{eq:channelG}). In the above equations, the  average power is noted under each  summand term.

\subsubsection{Treating interference as noise}
In terms of decoding, each intended receiver simply treats interference as noise. This implies that receiver~$k$ can  decode the codeword $v^{\bln}_k (w_k,  w_k')$  with arbitrarily small error probability when  $\bln$ gets large and  the rate of the codeword (i.e., $R_k + R_k'$) satisfies the following condition: 
\begin{align}
R_k + R_k' <  \Imu(v_k; y_k ) ,  \quad k=1,2      \label{eq:Rk876}
\end{align}
(cf.~\cite{CT:06}). 
Note that $R_k$ and $R_k'$ represent the rates of the secure message $w_k$ and the confusion message $w_k'$, respectively (cf.~\eqref{eq:code2341}).
Once  the codeword  $v^{\bln}_k (w_k,  w_k')$ is decoded, the message $w_k$ can be decoded directly from the codebook mapping.
Let us set
\begin{align}
R_k &\defeq    \Imu(v_k; y_k) -  \Imu ( v_k; y_{\ell} | v_{\ell} ) - \epsilon  \label{eq:Rk623} \\ 
R_k'  &\defeq   \Imu ( v_k; y_{\ell} | v_{\ell}) - \epsilon  \label{eq:Rk623b}  
\end{align}
for some $\epsilon >0$ and $k,\ell   \in \{1,2\}, k \neq \ell$. Obviously,  $R_k$ and $R_k'$ designed in \eqref{eq:Rk623} and \eqref{eq:Rk623b} satisfy the condition in \eqref{eq:Rk876}.

\subsubsection{Secure rate}
From the  proof of \cite[Theorem~2]{XU:15}   (or \cite[Theorem~2]{LMSY:08})  it implies that, given the above wiretap codebook and the rates designed in \eqref{eq:Rk623} and \eqref{eq:Rk623b},  the messages $w_1$ and $w_2$ are secure from their eavesdroppers, that is,  
$\Imu(w_1; y_{2}^{\bln})  \leq  \bln \epsilon$ and $\Imu(w_2; y_{1}^{\bln})  \leq  \bln \epsilon$. 
Therefore, by letting $\epsilon \to 0$, the scheme 
achieves the secure rate pair  $R_1  =  \Imu(v_1; y_1) -  \Imu ( v_1; y_2 | v_2 ) $ and $R_2  = \Imu(v_2; y_2) -  \Imu ( v_2; y_1 | v_1 )$.   Due to the Gaussian inputs  and  outputs  (see \eqref{eq:map888}, \eqref{eq:y441} and \eqref{eq:y442}), this achievable secure rate pair is expressed as 
 \begin{align}
R_1    &=  \underbrace{ \log \Bigl(  1+    \frac{P^{\alpha_{11} - \beta_1}}{ 1+ P^{\alpha_{1 2} - \beta_{2}}}  \Bigr)}_{= \Imu(v_1; y_1)} - \underbrace{ \log (1+ P^{\alpha_{21} - \beta_{1}})}_{  = \Imu ( v_1; y_2 | v_2 ) }   \non   \\
R_2   & =  \underbrace{\log \Bigl(  1+     \frac{P^{\alpha_{22} - \beta_2}}{ 1+ P^{\alpha_{21} - \beta_{1}}}\Bigr) }_{ = \Imu(v_2; y_2)  } -  \underbrace{ \log (1+ P^{\alpha_{12} - \beta_{2}}) }_{ = \Imu ( v_2; y_1 | v_1 ) }  \non
\end{align}
for some $\beta_1, \beta_2 \geq 0$.  By setting $\beta_1 =  \alpha_{21}$ and  $\beta_2 =  \alpha_{12}$, then the interference at each receiver is scaled down to the noise level (see~\eqref{eq:y441} and \eqref{eq:y442}) and the  achievable secure rate pair becomes
 \begin{align}
R_1    &=   \log \bigl(  1+    \frac{P^{\alpha_{11} - \alpha_{21} }}{2}  \bigr)  - 1    \label{eq:Rk111}    \\
R_2   & =  \log \bigl(  1+     \frac{P^{\alpha_{22} -  \alpha_{12} }}{ 2 }\bigr)  -  1.   \label{eq:Rk432}  
\end{align}
Note that in this GWC-TIN scheme, the transmitters do \emph{not} need to know  the channel  phases.

\section{Main results  \label{sec:mainresult}}

This section provides the main results for the  deterministic channel  model and Gaussian channel model, respectively.

\subsection{Results for the deterministic channel  \label{sec:mainresultDet}}

For the  deterministic  channel defined in Section~\ref{sec:systemDeterministic}, we identify a regime where a simple scheme without using cooperative jamming (i.e., WoCJ scheme, described in  Section~\ref{sec:noCJdet})  achieves the optimal secure sum capacity.  The result is stated in the following theorem.

 \vspace{5pt}
 
\begin{theorem}  \label{thm:capacitydetasym}
For the two-user  deterministic interference channel defined in Section~\ref{sec:systemDeterministic}, where $m_{k\ell} $ denotes the level of bits of the channel from transmitter~$\ell$ to receiver~$k$, $\forall k, \ell \in \{1,2\}$, if the following  conditions are satisfied,  
\begin{align}
    m_{22} + (m_{11 } - m_{12})^+       &\geq   m_{21}+  m_{12}     \label{eq:capdetcond1} \\
    m_{11} + (m_{22 } - m_{21})^+      & \geq   m_{21}+  m_{12}     \label{eq:capdetcond2}
\end{align}
then a simple scheme without using cooperative jamming (described in Section~\ref{sec:noCJdet}) achieves the optimal secure sum capacity, which  is 
\begin{align}
 C_{\text{sum}}  &  =  (m_{22 } - m_{12})^+     +    (m_{11 } - m_{21})^+ .     \label{eq:capdetasy2}
 \end{align}
\end{theorem}
 \vspace{3pt}
\begin{proof}
As discussed in Section~\ref{sec:noCJdet}, WoCJ scheme achieves a secure sum rate of $(m_{22 } - m_{12})^+     +    (m_{11 } - m_{21})^+ $ without using cooperative jamming.    
To prove its optimality, we provide an outer bound on the secure  capacity region of the deterministic channel, given in  Lemma~\ref{lm:detconverse} (see  Section~\ref{sec:converseDet}). The derived  outer bound  reveals that the secure sum rate achieved by WoCJ scheme is indeed optimal if the conditions in \eqref{eq:capdetcond1} and \eqref{eq:capdetcond2} are satisfied.  
In Remark~\ref{rk:detsy} (see  Section~\ref{sec:converseDet}), we show how to prove  Theorem~\ref{thm:capacitydetasym} by using the derived outer bound.
\end{proof}

 \vspace{3pt}

\begin{example}  \label{exp:det}
To interpret the result in Theorem~\ref{thm:capacitydetasym}, we consider  a setting with $m_{11}=m_{22}= 3, m_{21}=2$ and $m_{12}=1$. For this setting the conditions in \eqref{eq:capdetcond1} and \eqref{eq:capdetcond2} are satisfied. This implies from Theorem~\ref{thm:capacitydetasym}  that  WoCJ scheme achieves the optimal secure sum capacity  without using cooperative jamming. For this setting the  optimal secure sum capacity is characterized as $C_{\text{sum}} = 3$ bits/channel use.
\end{example}

 \vspace{3pt}
 
For the symmetric case with $\md  =  m_{11}  =  m_{22}$,  $\mc  =  m_{12}   =  m_{21}$, and $\alpha = \frac{\mc}{\md}$ (see \eqref{eq:detsym} and \eqref{eq:detsyma}),  Theorem~\ref{thm:capacitydetasym} reveals that if the following condition is satisfied,
\begin{align}
 0 \leq  \alpha \leq  2/3    \label{eq:symcond}
 \end{align}
 then WoCJ scheme  achieves the optimal secure sum capacity,  $ C_{\text{sum}}   = 2  (\md - \mc)$, without using cooperative jamming. 
More interestingly, for this symmetric case,  condition \eqref{eq:symcond} is indeed \emph{sufficient} and \emph{necessary}  for WoCJ scheme to be optimal in terms of secure sum capacity\footnote{When we determine whether a condition is necessary for a scheme to be optimal, we just focus on the regime where the secure sum capacity is strictly positive.}.  A more general result on the  symmetric case is stated in the following theorem.

\begin{figure}[t!]
\centering
\includegraphics[width=8.9cm]{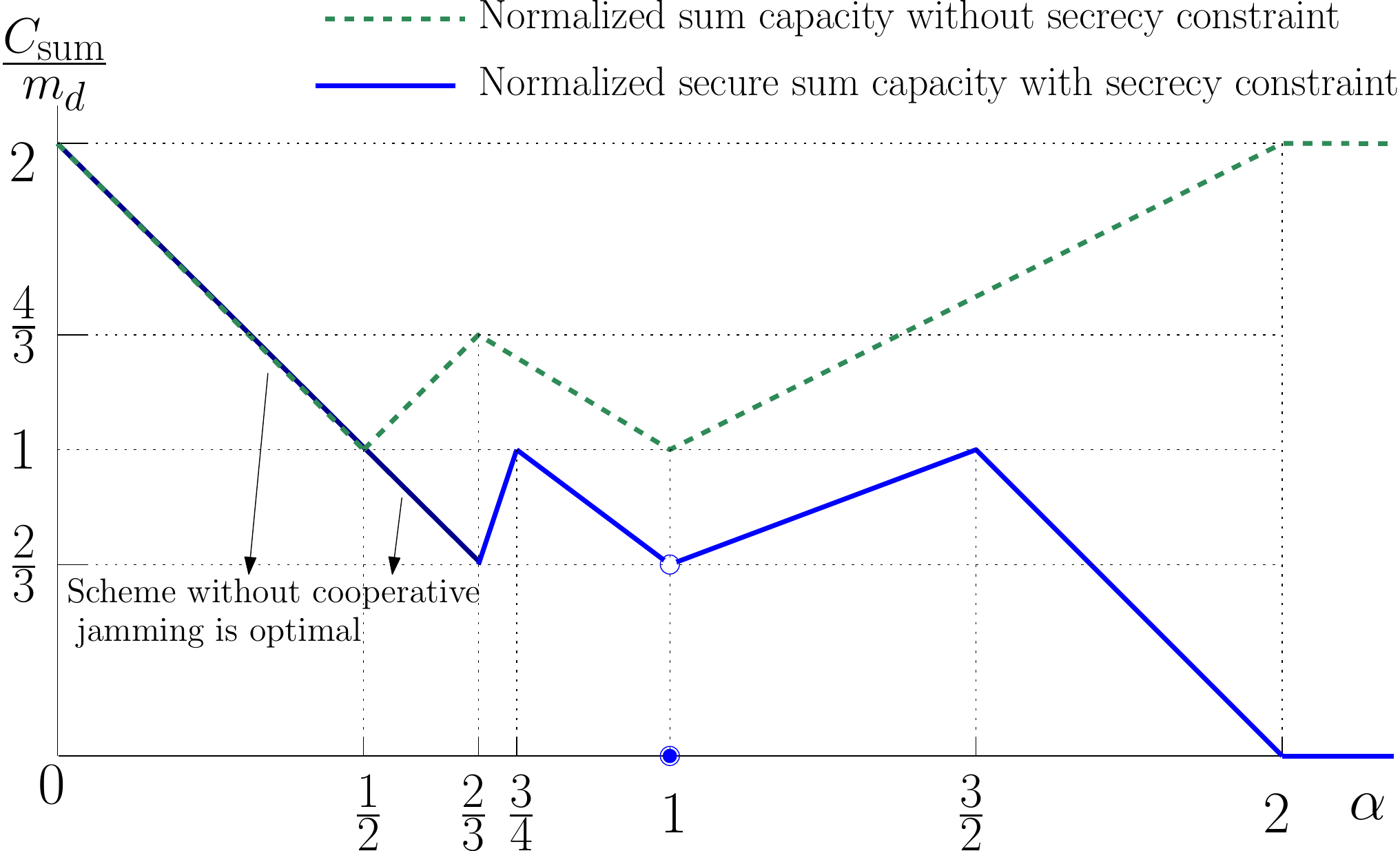}
\caption{Normalized  sum capacity vs. $\alpha$ for the two-user \emph{symmetric} deterministic interference channel with and without secrecy constraints, where $\alpha = \frac{\mc}{\md}$.  Note that a simple scheme  without using cooperative jamming achieves the optimal secure sum capacity if and only if  $\alpha \in [0,  \frac{2}{3}]$.}
\label{fig:Scapacity}
\end{figure}

\begin{theorem}  \label{thm:capacitydet}
For the two-user symmetric  deterministic interference channel defined in Section~\ref{sec:systemDeterministic}, the  normalized  secure sum capacity $\frac{C_{\text{sum}}}{\md}$ is characterized as 
\begin{subnumcases}
{\frac{C_{\text{sum}}}{\md}  =} 
     2(1- \alpha)    &    for   \ $ 0 \leq \alpha \leq  \frac{2}{3}$        \label{thm:capacitydet1} \\
    2(2\alpha- 1)  &  for \ $\frac{2}{3}  \leq \alpha \leq  \frac{3}{4}$    \label{thm:capacitydet2} \\ 
        2(1 -  \frac{2\alpha}{3})  &  for  \   $\frac{3}{4}  \leq \alpha <  1$    \label{thm:capacitydet3} \\ 
               0   &  for  \  $\alpha =  1$                                         \label{thm:capacitydet4}  \\ 
            \frac{2\alpha}{3}  &  for  \  $1  < \alpha \leq  \frac{3}{2}$    \label{thm:capacitydet5} \\ 
                        2(2- \alpha)  &  for \   $\frac{3}{2} \leq  \alpha \leq 2$   \label{thm:capacitydet6} \\ 
                                               0  &  for  \   $2\leq  \alpha$ .  \label{thm:capacitydet7}
\end{subnumcases}
Moreover,  a simple scheme without using cooperative jamming, that is, WoCJ scheme, achieves the optimal secure sum capacity if and only if  $\alpha \in [0,  \frac{2}{3}]$.
\end{theorem}

\begin{proof}
The converse is relegated to Section~\ref{sec:converseDet} (see Lemma~\ref{lm:detconverse} and Remark~\ref{rk:detasy} in Section~\ref{sec:converseDet}). 
For the achievability, note that WoCJ scheme  achieves a secure sum rate of $2  (\md - \mc)$ for this symmetric setting, again, \emph{without using cooperative jamming}. The secure sum rate achieved by WoCJ  scheme  is optimal when  $\alpha \in [0,  \frac{2}{3}]$ --- the optimality is proved through the derived converse.
For the other regimes with strictly positive secure sum capacity, i.e., $\alpha \in  (\frac{2}{3}, 1) \cup (1 , 2)$,  the secure sum rate achieved by WoCJ scheme is \emph{not} optimal. When $\alpha \in  (\frac{2}{3}, 1) \cup (1 , 2)$, the schemes \emph{with cooperative jamming} have been shown in \cite{GTJ:15} to achieve the optimal secure sum  capacity of this symmetric setting (see Appendix~\ref{sec:achisome} for the sketch of the schemes).
\end{proof}

Fig.~\ref{fig:Scapacity}  depicts the  normalized secure sum capacity  with  secrecy constraint (cf.~Theorem~\ref{thm:capacitydet}), as well as the normalized sum capacity without  secrecy constraint (cf.~\cite{ETW:08}), for the two-user \emph{symmetric} deterministic interference channel.  
Note that WoCJ scheme achieves the optimal secure sum capacity \emph{if and only if}  $\alpha \in [0,  \frac{2}{3}]$.
Also note that secrecy constraint  incurs no capacity penalty \emph{if and only if} $\alpha \in [0,  \frac{1}{2}]$.

\subsection{Results for the Gaussian channel  \label{sec:mainresultG}}

Let us now focus on the Gaussian channel defined in Section~\ref{sec:systemGaussian}.
For this  Gaussian channel,  we identify a regime where  the simple scheme without cooperative jamming, that is, GWC-TIN scheme described in  Section~\ref{sec:noCJGau}, achieves the optimal secure sum capacity to within a constant gap.   The result is stated in the following theorem.

\begin{theorem}  \label{thm:GaussianNCJ}
For the two-user  Gaussian interference channel defined in Section~\ref{sec:systemGaussian}, where $\alpha_{k\ell} $ denotes the channel strength from transmitter~$\ell$ to receiver~$k$, $\forall k, \ell \in \{1,2\}$, if the following the conditions are satisfied,  
\begin{align}
    \alpha_{22} + (\alpha_{11 } - \alpha_{12})^+       &\geq   \alpha_{21}+  \alpha_{12}     \label{eq:capGaussian1} \\
    \alpha_{11} + (\alpha_{22 } - \alpha_{21})^+      & \geq   \alpha_{21}+  \alpha_{12}     \label{eq:capGaussian2} 
\end{align}
then the simple scheme of using Gaussian wiretap codebook at each transmitter (without using cooperative jamming) and treating interference as noise at each intended receiver (that is, GWC-TIN scheme) achieves the optimal secure sum capacity to within a constant gap. 
Moreover, given the conditions of \eqref{eq:capGaussian1} and \eqref{eq:capGaussian2}, the optimal secure sum capacity $C_{\text{sum}}$ satisfies  
 \begin{align}
   \log \bigl(  1+    \frac{P^{\alpha_{11} - \alpha_{21} }}{2}  \bigr)  +  \log \bigl(  1+     \frac{P^{\alpha_{22} -  \alpha_{12} }}{ 2 }\bigr)  -  2 \leq  C_{\text{sum}}    \leq   \log (1+ 2 P^{\alpha_{22}-\alpha_{12}}) +   \log (1+ 2P^{\alpha_{11}-\alpha_{21}})  + 4, \label{eq:capa255}  
\end{align}
where the lower bound is achieved by GWC-TIN scheme. 
\end{theorem}
 \vspace{3pt}
\begin{proof}
As discussed in Section~\ref{sec:noCJGau}, GWC-TIN scheme achieves a secure sum rate of  $R_1 + R_2 = \log \bigl(  1+    \frac{P^{\alpha_{11} - \alpha_{21} }}{2}  \bigr)  +  \log \bigl(  1+     \frac{P^{\alpha_{22} -  \alpha_{12} }}{ 2 }\bigr)  -  2$  (see \eqref{eq:Rk111} and \eqref{eq:Rk432}). 
To prove the optimality of GWC-TIN scheme, we provide an upper bound on the secure  sum capacity  of the Gaussian channel, given in  Lemma~\ref{lm:gupper} (see  Section~\ref{sec:converse}). The derived  upper bound  reveals that, if the conditions in \eqref{eq:capGaussian1} and \eqref{eq:capGaussian2} are satisfied,  then the achievable secure sum rate of GWC-TIN scheme  indeed approaches the secure sum capacity to within a constant gap.  
In Remark~\ref{rk:gauas} (see  Section~\ref{sec:converse}), we show how to prove  Theorem~\ref{thm:GaussianNCJ} by using the derived upper bound.
\end{proof}

 \vspace{3pt}

\begin{example}  \label{exp:detGau}
To interpret the result in Theorem~\ref{thm:GaussianNCJ}, we consider  a setting with $\alpha_{11}=2, \alpha_{22}= 3, \alpha_{21}=1$ and $\alpha_{12}=2$. For this setting the conditions in \eqref{eq:capGaussian1} and \eqref{eq:capGaussian2} are satisfied. This implies from Theorem~\ref{thm:GaussianNCJ}  that  GWC-TIN scheme achieves the optimal secure sum capacity to within a constant gap,  without using cooperative jamming. 
\end{example}

 \vspace{5pt}

In \eqref{eq:capa255}, a secure sum capacity is characterized to within a constant gap. This directly implies the characterization of the secure sum GDoF.   
The following GDoF result   is concluded from Theorem~\ref{thm:GaussianNCJ}.
 \vspace{1pt}
\begin{corollary}[GDoF result] \label{corr:GDoF}
For the two-user  Gaussian interference channel,  if the conditions of \eqref{eq:capGaussian1} and \eqref{eq:capGaussian2} are satisfied,  then GWC-TIN scheme  achieves the optimal secure sum GDoF, which  is 
\begin{align}
 d_{\text{sum}}^{\star} &  = (\alpha_{22 } - \alpha_{12})^+     +    (\alpha_{11 } - \alpha_{21})^+  .  \label{eq:GDoFGaussianasy2}
 \end{align}
\end{corollary}

 \vspace{8pt}

For the \emph{symmetric} case with $\alpha_{11}  =  \alpha_{22}$ and $\alpha_{12}   =  \alpha_{21}$,  Theorem~\ref{thm:GaussianNCJ} directly implies that if the following condition is satisfied,
\begin{align}
 0 \leq  \frac{\alpha_{12}}{\alpha_{11}}  \leq  \frac{2}{3}    \label{eq:symcondGau}
 \end{align}
 then GWC-TIN scheme  achieves the optimal secure sum capacity to within a constant gap.   The result on the  symmetric case is stated in the following corollary.

 \vspace{3pt}

\begin{corollary} [Symmetric Gaussian channel] \label{cor:sym}
For the two-user  symmetric Gaussian  interference channel  with $\alpha_{11}  =  \alpha_{22}$ and $\alpha_{12}   =  \alpha_{21}$, if the condition of  $0 \leq  \frac{\alpha_{12}}{\alpha_{11}}  \leq  \frac{2}{3}$ is satisfied,  then GWC-TIN scheme  achieves the optimal secure sum capacity to within a constant gap.  
Under this condition, the optimal secure sum capacity $C_{\text{sum}}$ satisfies  
\[
 2 \log \bigl(  1+    \frac{P^{\alpha_{11} - \alpha_{21} }}{2}  \bigr)    -  2    \leq  C_{\text{sum}}    \leq  2 \log \bigl(  1+   2P^{\alpha_{11} - \alpha_{21} }  \bigr)  + 4, \]
while the optimal secure sum GDoF is characterized by 
\[
 d_{\text{sum}}^{\star}  = 2 (\alpha_{11 } - \alpha_{21}). 
\]
\end{corollary}

 \vspace{3pt}

\section{Converse for the deterministic channel  \label{sec:converseDet}}

For the deterministic channel defined in Section~\ref{sec:systemDeterministic}, we provide a  general outer bound on the secure  capacity region, which is stated in the following  lemma.

\begin{lemma}  \label{lm:detconverse}
For the two-user deterministic interference channel defined in Section~\ref{sec:systemDeterministic},  an  outer bound of secure  capacity region is given  by 
\begin{align}
R_1  & \leq     \max\{ 0,  \  m_{11}  -  ( m_{21} - m_{22} )^+ \}    \label{eq:detup1} \\ 
R_2 & \leq      \max\{ 0, \   m_{22}  -  ( m_{12} - m_{11} )^+ \}       \label{eq:detup2} \\ 
R_1 + R_2  & \leq    \max\{  m_{21}- (m_{11 } - m_{12})^+ ,  \  m_{22}- m_{12}, \ 0  \}     \non\\
 &  \quad  + \max\{  m_{12}- (m_{22 } - m_{21})^+ ,  \  m_{11}- m_{21},  \ 0 \}     \label{eq:detup3} \\ 
2R_1 + R_2  & \leq   \max\{ m_{11}, m_{12} \}  +  (  m_{11}  - m_{21})^+ +    (  m_{22}  - m_{12})^+   \label{eq:detup4} \\ 
2R_2 + R_1  & \leq    \max\{ m_{22}, m_{21} \}  +  (  m_{11}  - m_{21})^+ +    (  m_{22}  - m_{12})^+.   \label{eq:detup5}  
\end{align}
\end{lemma}

 \vspace{5pt}

Before describing the proof of Lemma~\ref{lm:detconverse},  we at first show how to prove Theorems~\ref{thm:capacitydetasym} and \ref{thm:capacitydet} by using the result of Lemma~\ref{lm:detconverse}. 

\begin{remark} [Converse proof of Theorem~\ref{thm:capacitydetasym}] \label{rk:detsy}
The converse of Theorem~\ref{thm:capacitydetasym} follows from bound~\eqref{eq:detup3}.   
Let us consider the two conditions in \eqref{eq:capdetcond1} and  \eqref{eq:capdetcond2}, i.e., $m_{22} + (m_{11 } - m_{12})^+  \geq    m_{21}+  m_{12} $ and $ m_{11} + (m_{22 } - m_{21})^+   \geq   m_{21}+  m_{12} $. 
 If these two conditions are satisfied,  then bound~\eqref{eq:detup3} can be rewritten as  
\begin{align}
R_1 + R_2   & \leq    (m_{22}- m_{12})^+  +    (m_{11}- m_{21})^+ .   \label{eq:detupext3} 
\end{align}
On the other hand, WoCJ scheme achieves a secure sum rate of $R_1 + R_2 = (m_{22 } - m_{12})^+     +    (m_{11 } - m_{21})^+ $  (see \eqref{eq:NoCJdetrate1} and \eqref{eq:NoCJdetrate2}), which matches the  upper bound derived in \eqref{eq:detupext3}. Therefore, 
WoCJ scheme achieves the optimal secure sum capacity if the conditions in \eqref{eq:capdetcond1} and \eqref{eq:capdetcond2} are satisfied.
\end{remark}

 \vspace{5pt}

\begin{remark} [Converse proof of Theorem \ref{thm:capacitydet}] \label{rk:detasy}
Note that Theorem~\ref{thm:capacitydet} is limited to the symmetric case with $\md  =  m_{11}  =  m_{22}$ and  $\mc  =  m_{12}   =  m_{21}$.   The converse of  \eqref{thm:capacitydet1} and \eqref{thm:capacitydet2} in Theorem~\ref{thm:capacitydet} follows from bound~\eqref{eq:detup3}. The converse of \eqref{thm:capacitydet3} - \eqref{thm:capacitydet7} in Theorem~\ref{thm:capacitydet} has been proved in \cite{MM:14o} that considered the two-user symmetric deterministic interference channel with transmitter cooperation.  
The result of \cite{MM:14o} has been extended in this work to the general setting of the deterministic interference channel, stated in \eqref{eq:detup1}, \eqref{eq:detup2}, \eqref{eq:detup4} and \eqref{eq:detup5}. 
\end{remark}

 \vspace{8pt}

As mentioned in Remarks~\ref{rk:detsy} and \ref{rk:detasy}, bound \eqref{eq:detup3} in Lemma~\ref{lm:detconverse}  proves the optimality of the schemes without cooperative jamming in the declared regimes (see~Theorems~\ref{thm:capacitydetasym}  and \ref{thm:capacitydet}).
In what follows we will provide the proof of \eqref{eq:detup3}, while the proofs of \eqref{eq:detup1}, \eqref{eq:detup2}, \eqref{eq:detup4} and \eqref{eq:detup5}  are relegated to Appendix~\ref{sec:conversesome}.

 \subsection{Proof of bound \eqref{eq:detup3} \label{sec:detup3} } 

For the proof of bound \eqref{eq:detup3}, our approach is different from the genie-aided approach that is commonly used in the settings without secrecy constraints (cf.~\cite{ETW:08}). In the genie-aided approach,  some genie-aided information is typically provided to the receivers, which might give a loose bound in our setting.

\begin{figure}[t!]
\centering
\includegraphics[width=13cm]{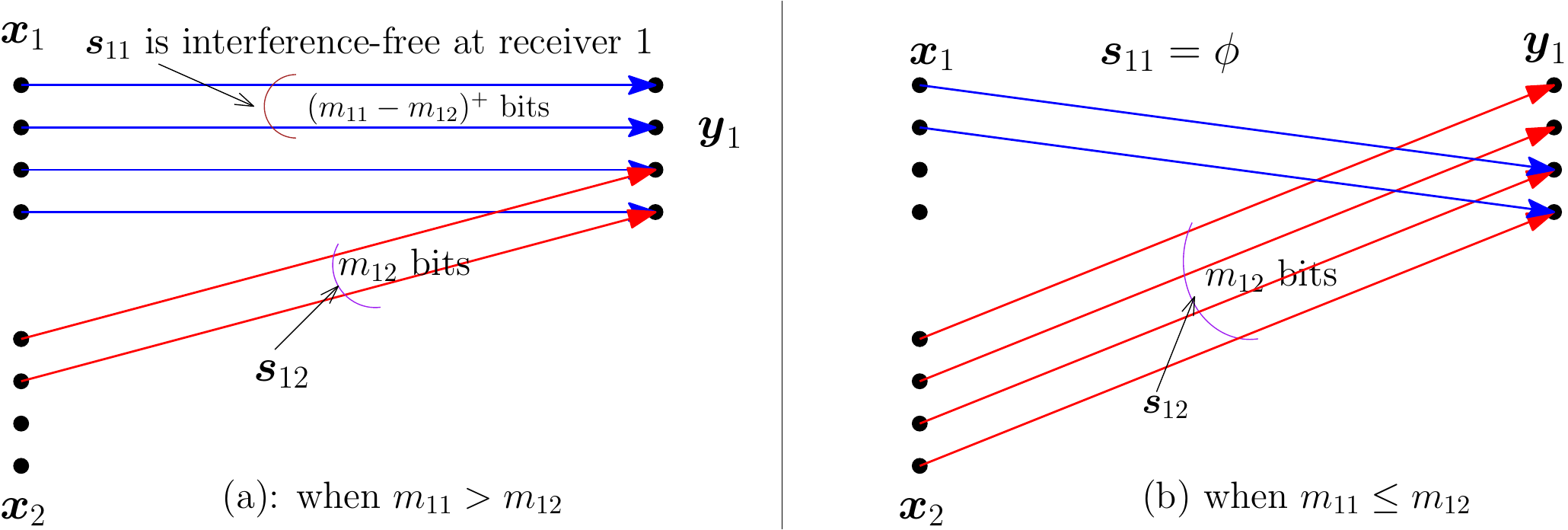}
\caption{Deterministic channel model for receiver~1.  $\sv_{11}$ represents the top $(m_{11} - m_{12})^+ $  bits of $\xv_1$, i.e., the interference-free bits sent from transmitter~1 to receiver~1.  $\sv_{12}$ represents the top $m_{12}$ bits of $\xv_2$, i.e.,  the signal bits sent from transmitter~2 to receiver~1.} 
\label{fig:ICdet}
\end{figure}

At first, let  \[ \xv_{k,1: \tau} (t)  \defeq  [x_{k,1}(t), x_{k,2}(t), \cdots , x_{k, \tau }(t)]^\T \]
that  represents the top (most significant) $\tau$ bits of $\xv_{k} (t)$, and let $\xv^{\bln}_{k,1: \tau} \defeq \{ \xv_{k,1: \tau } (t)\}_{t=1}^{\bln}$, for some positive $\tau$ and $k=1,2$.
We  also define that 
\begin{align}
  \sv_{11}(t)  & \defeq  \xv_{1,1: (m_{11} - m_{12})^+}(t)    \label{eq:defs11det} \\
 \sv_{12}(t)  &\defeq  \xv_{2, 1: m_{12}}  (t) .  \label{eq:defs12det}
\end{align}
In our context, $\sv_{11}(t)$ represents the top $(m_{11} - m_{12})^+ $ bits of $\xv_1(t)$, while $\sv_{12}(t)$ represents the top $ m_{12} $ bits of $\xv_2(t)$ (see Fig.~\ref{fig:ICdet}).
We begin with the rate of user~1:
\begin{align}
  \bln R_1 &= \Hen(w_1)    \nonumber \\
  &= \Imu(w_1; \yv^{\bln}_1) + \Hen(w_1| \yv^{\bln}_1)  \nonumber \\
  &\leq  \Imu(w_1; \yv^{\bln}_1) + \bln \epsilon_{1,n}    \label{eq:Fanodet}  \\
  &\leq  \Imu(w_1;  \yv^{\bln}_1) - \Imu(w_1; \yv^{\bln}_2) +  \bln \epsilon_{1,n} + \bln \epsilon   \label{eq:secrecydet} 
\end{align}
where \eqref{eq:Fanodet} follows from  Fano's inequality,  $\lim_{n\to\infty} \epsilon_{1,n} = 0$;
\eqref{eq:secrecydet} results from secrecy constraint in \eqref{eq:defsecrecy1}, i.e.,  $ \Imu(w_1; \yv^{\bln}_2)  \leq \bln \epsilon$ for an  arbitrary small  $\epsilon$. Similarly, for the rate of user~2 we have 
\begin{align}
  \bln R_2 & \leq  \Imu(w_2; \yv^{\bln}_2) - \Imu(w_2 ; \yv^{\bln}_1) +  \bln \epsilon_{2,n}  + \bln \epsilon    \label{eq:secrecy2det} 
\end{align}
which, together with \eqref{eq:secrecydet}, gives the following bound on the sum rate:
\begin{align}
 & \bln R_1 + \bln R_2 -\bln \epsilon_{1,n} -\bln \epsilon_{2,n} -2 \bln \epsilon     \nonumber \\
  & \leq \Imu(w_1; \yv^{\bln}_1) - \Imu(w_1;  \yv^{\bln}_2)  +  \Imu(w_2 ;  \yv^{\bln}_2) - \Imu(w_2 ;  \yv^{\bln}_1)     \nonumber \\
  & = \Hen(\yv^{\bln}_1) - \Hen(\yv^{\bln}_1| w_1)  -    \Hen(\yv^{\bln}_2) +  \Hen(\yv^{\bln}_2| w_1)  \nonumber\\ &\quad + \Hen(\yv^{\bln}_2) - \Hen(\yv^{\bln}_2| w_2)  -    \Hen(\yv^{\bln}_1) +  \Hen(\yv^{\bln}_1| w_2)       \nonumber \\
  & = \Hen(\yv^{\bln}_2| w_1)  - \Hen(\yv^{\bln}_1| w_1)   +  \Hen(\yv^{\bln}_1| w_2)    - \Hen(\yv^{\bln}_2| w_2).     \label{eq:sumrate1det} 
\end{align}
For the first two terms in the right-hand side of \eqref{eq:sumrate1det}, we have:
\begin{align}
  &  \Hen(\yv^{\bln}_2| w_1)  - \Hen(\yv^{\bln}_1| w_1)       \nonumber \\
 &=  \Hen( \sv_{11}^{\bln}, \yv^{\bln}_2| w_1)  -   \underbrace{\Hen(\sv_{11}^{\bln}| \yv^{\bln}_2, w_1)}_{\defeq \J_1}   - \Hen(\sv_{11}^{\bln}, \yv^{\bln}_1| w_1)    + \underbrace{\Hen(\sv_{11}^{\bln} | \yv^{\bln}_1, w_1) }_{\defeq \J_2}  \label{eq:bd3141det}   \\ 
& =  \Hen(\sv_{11}^{\bln}| w_1)  +  \Hen( \yv^{\bln}_2 | \sv_{11}^{\bln}, w_1)     - \Hen(\sv_{11}^{\bln} | w_1) - \Hen( \yv^{\bln}_1|\sv_{11}^{\bln},  w_1)  - \J_1  + \J_2  \non  \\
 &=    \Hen( \yv^{\bln}_2| \sv_{11}^{\bln}, w_1) - \Hen( \yv^{\bln}_1|\sv_{11}^{\bln},  w_1)   - \J_1  + \J_2      \nonumber   \\
& =    \Hen(\sv_{12}^{\bln},  \yv^{\bln}_2| \sv_{11}^{\bln}, w_1) -  \underbrace{\Hen(\sv_{12}^{\bln} | \yv^{\bln}_2,  \sv_{11}^{\bln}, w_1)}_{\defeq \J_3}    - \Hen( \yv^{\bln}_1|\sv_{11}^{\bln},  w_1)   - \J_1  + \J_2  \label{eq:bd3743det}      \\
 &=    \Hen(\sv_{12}^{\bln}| \sv_{11}^{\bln}, w_1)  +   \underbrace{\Hen( \yv^{\bln}_2| \sv_{12}^{\bln}, \sv_{11}^{\bln}, w_1)}_{\defeq \J_4}   - \Hen( \yv^{\bln}_1|\sv_{11}^{\bln},  w_1)  - \J_1  + \J_2 -\J_3 \label{eq:bd8225det}      \\
  &=   \Hen(\yv^{\bln}_1| x_1^{\bln}, \sv_{11}^{\bln}, w_1)     -  \Hen( \yv^{\bln}_1|\sv_{11}^{\bln},  w_1)    -  \J_1   +  \J_2  - \J_3  + \J_4 \label{eq:bd2535det}      \\ 
&=    -  \Imu (\yv^{\bln}_1;  \xv_1^{\bln}| \sv_{11}^{\bln}, w_1)  - \J_1  + \J_2 - \J_3 +\J_4  \label{eq:bd8256det}   
 \end{align}
where  $\J_1\defeq \Hen(\sv_{11}^{\bln}| \yv^{\bln}_2, w_1)$, $ \J_2 \defeq \Hen(\sv_{11}^{\bln} | \yv^{\bln}_1, w_1) $, $\J_3 \defeq \Hen(\sv_{12}^{\bln} | \yv^{\bln}_2, \sv_{11}^{\bln}, w_1)$ and $\J_4\defeq \Hen( \yv^{\bln}_2| \sv_{12}^{\bln}, \sv_{11}^{\bln}, w_1)$;
 the steps from \eqref{eq:bd3141det} to \eqref{eq:bd8225det}   follow   from  chain rule; 
 \eqref{eq:bd2535det}  follows from that 
\begin{align}
    \Hen(\sv_{12}^{\bln}| \sv_{11}^{\bln}, w_1)    
 = &  \Hen(\sv_{12}^{\bln})   \label{eq:bd8325det}  \\
 = &  \Hen(\sv_{12}^{\bln}| \xv^{\bln}_1, \sv_{11}^{\bln}, w_1)   \label{eq:bd1794det}      \\
 = &  \Hen\bigl( \bigl\{   S^{q- m_{12}} \xv_2 (t) \bigr\}_{t=1}^{\bln} | \xv^{\bln}_1, \sv_{11}^{\bln}, w_1\bigr)  \nonumber   \\
 = &  \Hen\bigl( \bigl\{ S^{q- m_{11}} \xv_1(t)  \oplus   S^{q- m_{12}} \xv_2(t)  \bigr\}_{t=1}^{\bln} | \xv^{\bln}_1, \sv_{11}^{\bln}, w_1\bigr)   \label{eq:bd7420det}       \\
 = &  \Hen(\yv^{\bln}_1| \xv^{\bln}_1, \sv_{11}^{\bln}, w_1)    \nonumber
 \end{align}
where \eqref{eq:bd8325det} and \eqref{eq:bd1794det} use the fact that $\sv_{12}^{\bln}$ is independent of $\sv_{11}^{\bln}, w_1$ and $ \xv^{\bln}_1$;  \eqref{eq:bd7420det} follows from the fact that $\Hen(\av|\bv)=\Hen(\av\oplus \bv | \bv)$ for any random $\av, \bv  \in
\F_{2}^{q}$.    
Going back to \eqref{eq:bd8256det}, we further have
\begin{align}
    \Hen(\yv^{\bln}_2| w_1)  - \Hen(\yv^{\bln}_1| w_1)   & = -  \Imu (\yv^{\bln}_1;  \xv_1^{\bln}| \sv_{11}^{\bln}, w_1)  - \J_1  + \J_2 - \J_3 +\J_4  \label{eq:bd9295det} \\
  & \leq   \J_2 + \J_4   \label{eq:bd3325det} \\
&    \leq \bln \cdot \max\{  m_{21}- (m_{11 } - m_{12})^+ ,   m_{22}- m_{12}, 0  \}    \label{eq:bd92374det}   
 \end{align}
 where \eqref{eq:bd9295det} is from \eqref{eq:bd8256det};  \eqref{eq:bd3325det}  follows from the fact that mutual information and entropy are always nonnegative;
 \eqref{eq:bd92374det} follows from Lemma~\ref{lm:Jboundsdet} (see below).
Similarly, by interchanging the roles of user~1 and user~2, we also have 
\begin{align}
   \Hen(\yv^{\bln}_1| w_2)  - \Hen(\yv^{\bln}_2| w_2)  
   \leq \   \bln \cdot \max\{  m_{12}- (m_{22 } - m_{21})^+ ,  \  m_{11}- m_{21}, \ 0 \}   \label{eq:bd9256det}.     
 \end{align}
  Finally, combining \eqref{eq:sumrate1det}, \eqref{eq:bd92374det} and \eqref{eq:bd9256det} gives the following bound on the sum rate
 \begin{align}
 R_1 + R_2  \leq &  \max\{  m_{21}- (m_{11 } - m_{12})^+ ,  \  m_{22}- m_{12}, \ 0  \}  \non \\
 & +    \max\{  m_{12}- (m_{22 } - m_{21})^+ ,  \  m_{11}- m_{21},  \ 0 \}  + \epsilon_{1,n} \!+\! \epsilon_{2,n}\! +\! 2  \epsilon .  \nonumber 
\end{align}
Letting  $n\to \infty, \  \epsilon_{1,n} \to 0, \ \epsilon_{2,n} \to 0$ and  $\epsilon \to 0$, we  get the desired bound \eqref{eq:detup3}. 
The lemma used in our proof is provided below.

 \vspace{10pt}
\begin{lemma}  \label{lm:Jboundsdet}
For  $ \J_2 = \Hen(\sv_{11}^{\bln} | \yv^{\bln}_1, w_1) $ and $\J_4= \Hen( \yv^{\bln}_2| \sv_{12}^{\bln}, \sv_{11}^{\bln}, w_1)$, we have
\begin{align}
\J_2  & \leq  0     \label{eq:Jb2det}  \\
\J_4  & \leq  \bln \cdot  \max\{  m_{21}- (m_{11 } - m_{12})^+ , \   m_{22}- m_{12},  \ 0  \} .    \label{eq:Jb4det}  
\end{align}
\end{lemma}
\begin{proof}
Begin with  $\J_2$, we have:
\begin{align}
\J_2  & = \Hen(\sv_{11}^{\bln} | \yv^{\bln}_1, w_1)  \nonumber \\
        & = \sum_{t=1}^{\bln}  \Hen(\sv_{11}(t) | \sv_{11}^{t-1}, \yv^{\bln}_1, w_1)  \label{eq:Jb0263det}  \\
        & \leq  \sum_{t=1}^{\bln}  \Hen( \sv_{11}(t) |  \yv_1(t))  \label{eq:Jb9255det}  \\
        &= 0    \label{eq:Jb9015det}
\end{align}
where \eqref{eq:Jb0263det} results from chain rule;
 \eqref{eq:Jb9255det} uses the fact that conditioning reduces  entropy;    
  \eqref{eq:Jb9015det} follows from the fact that  $\sv_{11}(t)$ can be reconstructed from $\yv_1(t)$  (see Fig.~\ref{fig:ICdet}).

Now we focus on the upper bound of $\J_4$:
\begin{align}
 \J_4   
& = \Hen( \yv^{\bln}_2| \sv_{12}^{\bln}, \sv_{11}^{\bln}, w_1)  \nonumber \\
        & = \sum_{t=1}^{\bln}  \Hen(\yv_{2}(t) | \yv_{2}^{t-1}, \sv_{12}^{\bln}, \sv_{11}^{\bln}, w_1)  \label{eq:Jb02631det}  \\
        & \leq  \sum_{t=1}^{\bln}  \Hen(\yv_{2}(t) | \sv_{12}(t), \sv_{11}(t) )  \label{eq:Jb92551det}  \\
        & \leq   \bln \cdot  \max\{  m_{21}- (m_{11 } - m_{12})^+ ,  \  m_{22}- m_{12}, \ 0  \}    \label{eq:Jb0099det}
 \end{align}
where \eqref{eq:Jb02631det} results from chain rule; \eqref{eq:Jb92551det} follows from the fact that conditioning reduces entropy;  
\eqref{eq:Jb0099det} holds true because,  $\sv_{12}(t)$ and $\sv_{11}(t)$ allow for the reconstruction of  all the bits except   $\max\{  m_{21}- (m_{11 } - m_{12})^+ ,  \  m_{22}- m_{12}, \ 0  \}$ least significant  bits of $\yv_{2}(t) $.  As shown in Fig.~\ref{fig:ICdety2}, $\sv_{11}(t)$ represents the top $ (m_{11 } - m_{12})^+$ bits of $\xv_1(t)$, which indicates that $\sv_{11}(t)$ can reconstruct the top $ (m_{11 } - m_{12})^+$ bits of the cross-link signal sent from transmitter~1 to receiver~2. The   $(m_{21}- (m_{11 } - m_{12})^+)^+$ least significant bits of cross-link signal sent from transmitter~1 to receiver~2 could not be recovered from $\sv_{11}(t)$. Similarly, $\sv_{12}(t)$ can reconstruct the top $m_{12}$ bits, but not the  $(m_{22}- m_{12})^+$ least significant bits,  of the direct-link signal sent from transmitter~2 to receiver~2. Therefore, $\sv_{12}(t)$ and $\sv_{11}(t)$ can reconstruct partial bits of $\yv_{2}(t) $, leaving   $\max\{  (m_{21}- (m_{11 } - m_{12})^+)^+ ,  \  (m_{22}- m_{12})^+ \}$ least significant  bits of $\yv_{2}(t) $ unreconstructable.
\end{proof}

\begin{figure}[t!]
\centering
\includegraphics[width=7cm]{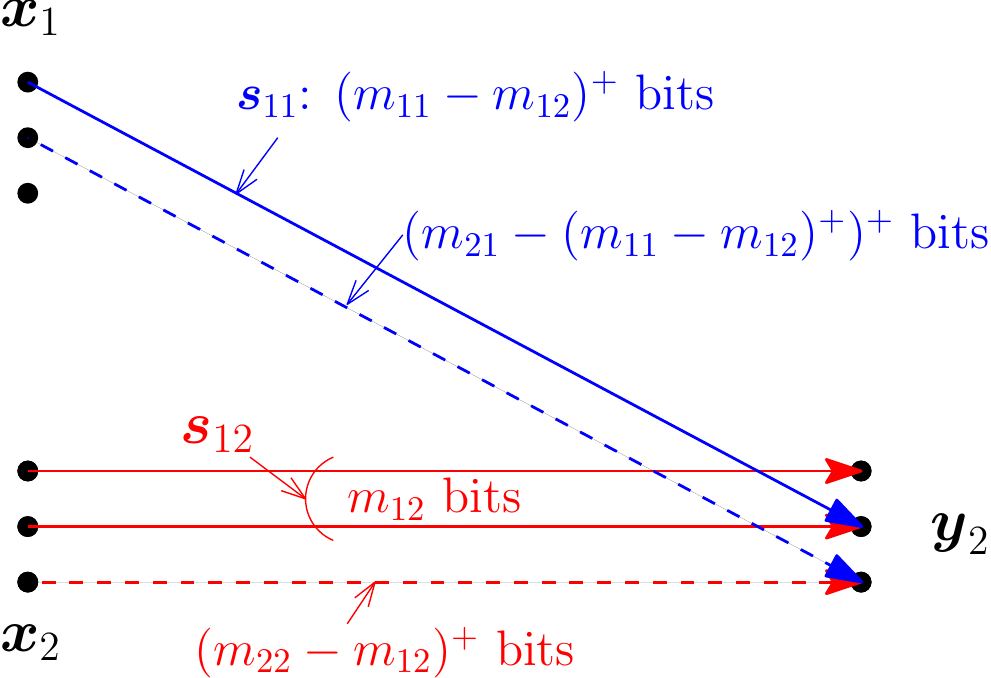}
\caption{Deterministic channel model for receiver~2.  The solid pipe lines are associated with the bits that can be reconstructed from either $\sv_{11}$ or $\sv_{12}$, while the  dash pipe lines are associated with the bits that can \emph{not} be reconstructed from  $\{\sv_{11}, \sv_{12}\}$. Therefore, $\sv_{11}$ and $\sv_{12}$ can reconstruct all the bits  except $\max\{  (m_{21}- (m_{11 } - m_{12})^+)^+ ,  \  (m_{22}- m_{12})^+ \}$ least significant  bits of  $\yv_{2}$.  For  example, when $m_{11}=m_{22} = 3$ and $m_{12}=m_{21} = 2$,  then $\sv_{11}$ and $\sv_{12}$ can reconstruct all the bits  except  one least significant bit of  $\yv_{2}$.} 
\label{fig:ICdety2}
\end{figure}

\section{Converse for the Gaussian channel \label{sec:converse}}

For the Gaussian interference channel defined in Section~\ref{sec:systemGaussian}, we  provide an upper bound on  the secure  sum capacity, which is stated in the following  lemma.

\begin{lemma}  \label{lm:gupper}
For the two-user Gaussian  interference channel defined in Section~\ref{sec:systemGaussian},  the  secure sum capacity is upper bounded by
\begin{align}
C_{\text{sum}}  &\leq   \log (1+ P^{\alpha_{22}-\alpha_{12}} + P^{\alpha_{22} - (\alpha_{11} - \alpha_{12})^+}) +   \log (1+ P^{\alpha_{11}-\alpha_{21}} + P^{\alpha_{11} - (\alpha_{22} - \alpha_{21})^+})   + 4 . \label{eq:gaussianupbound}  
\end{align}
\end{lemma}

 \vspace{5pt}

\begin{remark} [Converse proof of Theorem~\ref{thm:GaussianNCJ}] \label{rk:gauas}
The converse of Theorem~\ref{thm:GaussianNCJ} follows  from the upper bound in \eqref{eq:gaussianupbound}.  
Let us consider the two conditions in \eqref{eq:capGaussian1} and  \eqref{eq:capGaussian2}, i.e., $\alpha_{22} + (\alpha_{11 } - \alpha_{12})^+  \geq   \alpha_{21}+  \alpha_{12}$ and $ \alpha_{11} + (\alpha_{22 } - \alpha_{21})^+  \geq   \alpha_{21}+  \alpha_{12}  $. 
 If these two conditions are satisfied,  then we can further bound the secure sum capacity by  
\begin{align}
C_{\text{sum}}  &\leq   \log (1+ 2 P^{\alpha_{22}-\alpha_{12}}) +   \log (1+ 2P^{\alpha_{11}-\alpha_{21}})  + 4     \label{eq:gaussianupbound22} 
\end{align}
where \eqref{eq:gaussianupbound22} results from bound \eqref{eq:gaussianupbound}, as well as the facts that  $P^{\alpha_{22}-\alpha_{12}} \geq P^{\alpha_{22} - (\alpha_{11} - \alpha_{12})^+}$ and that $P^{\alpha_{11}-\alpha_{21}} \geq P^{\alpha_{11} - (\alpha_{22} - \alpha_{21})^+}$  under the two conditions in \eqref{eq:capGaussian1} and  \eqref{eq:capGaussian2}.
On the other hand, GWC-TIN scheme described in Section~\ref{sec:noCJGau} achieves a secure sum rate of 
\begin{align}
R_1 + R_2 = \log \bigl(  1+    \frac{P^{\alpha_{11} - \alpha_{21} }}{2}  \bigr)  +  \log \bigl(  1+     \frac{P^{\alpha_{22} -  \alpha_{12} }}{ 2 }\bigr)  -  2   \label{eq:GWCTINsum} 
\end{align}
(see \eqref{eq:Rk111} and \eqref{eq:Rk432}).  
One can easily check that the gap between  the secure sum capacity lower bound in \eqref{eq:GWCTINsum} and upper bound in \eqref{eq:gaussianupbound22}  is no more than $10$ bits/channel use. Note that, the gap can be further reduced by optimizing the computations in converse and achievability.  
\end{remark}

 \vspace{9pt}

Let us now prove Lemma~\ref{lm:gupper}.  The proof follows closely from that derived for the deterministic case (see~Section~\ref{sec:detup3}). 
At first we define that 
\begin{align}
  s_{11}(t)  & \defeq   \sqrt{P^{(\alpha_{11}-\alpha_{12})^+}} e^{j\theta_{11}} x_{1}(t) +  \tilde{z}_{1}(t)    \label{eq:defs11} \\
 s_{12}(t)  & \defeq   \sqrt{P^{\alpha_{12}}} e^{j\theta_{12}} x_{2}(t)  +z_{1}(t),  \label{eq:defs12}
\end{align}
where  $\tilde{z}_{1}(t) \sim \mathcal{C}\mathcal{N}(0, 1)$ is a virtual noise that is independent of the other noise and transmitted signals.
By following the steps of \eqref{eq:Fanodet} - \eqref{eq:sumrate1det} derived in Section~\ref{sec:detup3}, we can bound the sum rate of this Gaussian channel  as
\begin{align}
  \bln R_1 + \bln R_2 -\bln \epsilon_{1,n} -\bln \epsilon_{2,n} -2 \bln \epsilon    
    \leq & \  \hen(y^{\bln}_2| w_1)  - \hen(y^{\bln}_1| w_1)   +  \hen(y^{\bln}_1| w_2)    - \hen(y^{\bln}_2| w_2).   \label{eq:sumrate1} 
 \end{align}
Then, following the steps of \eqref{eq:bd3141det} - \eqref{eq:bd8256det} gives:
\begin{align}
  &  \hen(y^{\bln}_2| w_1)  - \hen(y^{\bln}_1| w_1)       \nonumber \\
  &=  \hen(s_{11}^{\bln}, y^{\bln}_2| w_1)  -   \underbrace{\hen(s_{11}^{\bln}| y^{\bln}_2, w_1)}_{\defeq J_1}   - \hen(s_{11}^{\bln}, y^{\bln}_1| w_1)    + \underbrace{\hen(s_{11}^{\bln} | y^{\bln}_1, w_1) }_{\defeq J_2}  \label{eq:bd3141}   \\ 
 &=    \hen( y^{\bln}_2| s_{11}^{\bln}, w_1) - \hen( y^{\bln}_1|s_{11}^{\bln},  w_1)   - J_1  + J_2      \nonumber   \\
& =   \hen(s_{12}^{\bln},  y^{\bln}_2| s_{11}^{\bln}, w_1) -  \underbrace{\hen(s_{12}^{\bln} | y^{\bln}_2,  s_{11}^{\bln}, w_1)}_{\defeq J_3}    - \hen( y^{\bln}_1|s_{11}^{\bln},  w_1)   - J_1  + J_2 \non \\
 &=    \hen(s_{12}^{\bln}| s_{11}^{\bln}, w_1)  +   \underbrace{\hen( y^{\bln}_2| s_{12}^{\bln}, s_{11}^{\bln}, w_1)}_{\defeq J_4}   - \hen( y^{\bln}_1|s_{11}^{\bln},  w_1)   - J_1  + J_2 -J_3 \label{eq:bd8225}      \\
  &=   \hen(y^{\bln}_1| x_1^{\bln}, s_{11}^{\bln}, w_1)     -  \hen( y^{\bln}_1|s_{11}^{\bln},  w_1)    -  J_1   +  J_2  - J_3  + J_4 \label{eq:bd2535}      \\ 
 & =   -  \Imu (y^{\bln}_1;  x_1^{\bln}| s_{11}^{\bln}, w_1)  - J_1  + J_2 -J_3 +J_4 \label{eq:bd992211}  
 \end{align}
where  $J_1\defeq \hen(s_{11}^{\bln}| y^{\bln}_2, w_1)$, $ J_2 \defeq \hen(s_{11}^{\bln} | y^{\bln}_1, w_1) $, $J_3 \defeq \hen(s_{12}^{\bln} | y^{\bln}_2, s_{11}^{\bln}, w_1)$ and $J_4\defeq \hen( y^{\bln}_2| s_{12}^{\bln}, s_{11}^{\bln}, w_1)$;
  the steps from \eqref{eq:bd3141} to \eqref{eq:bd8225}   follow from   chain rule;
  \eqref{eq:bd2535}  stems from that 
\begin{align}
    \hen(s_{12}^{\bln}| s_{11}^{\bln}, w_1)     
 = &  \hen(s_{12}^{\bln})   \label{eq:bd8325}  \\
 = &  \hen(s_{12}^{\bln}| x^{\bln}_1, s_{11}^{\bln}, w_1)   \label{eq:bd1794}      \\
 = &  \hen( \{ \sqrt{P^{\alpha_{12}}} e^{j\theta_{12}} x_{2}(t)  +z_{1}(t)  \}_{t=1}^{\bln} | x^{\bln}_1, s_{11}^{\bln}, w_1)  \nonumber   \\
 = &  \hen( \{    \sqrt{P^{\alpha_{11}}} e^{j\theta_{11}} x_{1}(t)  + \sqrt{P^{\alpha_{12}}} e^{j\theta_{12}} x_{2}(t)  +z_{1}(t)  \}_{t=1}^{\bln} | x^{\bln}_1, s_{11}^{\bln}, w_1)   \label{eq:bd7420}       \\
 = &  \hen(y^{\bln}_1| x_1^{\bln}, s_{11}^{\bln}, w_1)    \nonumber
 \end{align}
where \eqref{eq:bd8325} and \eqref{eq:bd1794} use the fact that $s_{12}^{\bln}$ is independent of $s_{11}^{\bln}, w_1$ and $ x^{\bln}_1$;
\eqref{eq:bd7420} follows from that $\hen(a|b)=\hen(a+b |b)$ for any continuous random variables $a$ and $b$.    
Going back to \eqref{eq:bd992211}, we further have 
\begin{align}
   \hen(y^{\bln}_2| w_1)  - \hen(y^{\bln}_1| w_1)  & =  -  \Imu (y^{\bln}_1;  x_1^{\bln}| s_{11}^{\bln}, w_1)  - J_1  + J_2 -J_3 +J_4     \label{eq:bd887733} \\
   & \leq  - J_1  + J_2 -J_3 +J_4     \label{eq:bd33405} \\
   &\leq    n\log (1+ P^{\alpha_{22}-\alpha_{12}} + P^{\alpha_{21} - (\alpha_{11} - \alpha_{12})^+}) + n\log 4  \label{eq:bd92374}  
 \end{align}
 where \eqref{eq:bd887733} is from \eqref{eq:bd992211};  
 \eqref{eq:bd33405}   stems from  the nonnegativity of mutual information; 
 \eqref{eq:bd92374} follows from Lemma~\ref{lm:Jbounds} (see below).
Similarly, by interchanging the roles of user~1 and user~2, we also have 
\begin{align}
   \hen(y^{\bln}_1| w_2)  - \hen(y^{\bln}_2| w_2) 
   \leq      n\log (1+ P^{\alpha_{11}-\alpha_{21}} + P^{\alpha_{12} - (\alpha_{22} - \alpha_{21})^+}) + n\log 4  \label{eq:bd9256}.     
 \end{align}
  Finally, combining \eqref{eq:sumrate1}, \eqref{eq:bd92374} and \eqref{eq:bd9256} gives the following bound on the sum rate
 \begin{align}
 R_1 + R_2   
 &\leq   \log (1+ P^{\alpha_{22}-\alpha_{12}} + P^{\alpha_{22} - (\alpha_{11} - \alpha_{12})^+}) + 2     \non \\
 &\quad +   \log (1+ P^{\alpha_{11}-\alpha_{21}} + P^{\alpha_{11} - (\alpha_{22} - \alpha_{21})^+}) + 2  + \epsilon_{1,n} + \epsilon_{2,n} + 2  \epsilon .
\end{align}
 By setting  $n\to \infty, \  \epsilon_{1,n} \to 0, \ \epsilon_{2,n} \to 0$ and  $\epsilon \to 0$, we  get the desired bound \eqref{eq:gaussianupbound}. 
The lemma used in our proof is provided below.

 \vspace{10pt}
\begin{lemma}  \label{lm:Jbounds}
For $J_1 = \hen(s_{11}^{\bln}| y^{\bln}_2, w_1)$, $ J_2 = \hen(s_{11}^{\bln} | y^{\bln}_1, w_1) $, $J_3 = \hen(s_{12}^{\bln} | y^{\bln}_2, s_{11}^{\bln}, w_1)$ and $J_4= \hen( y^{\bln}_2| s_{12}^{\bln}, s_{11}^{\bln}, w_1)$,  we have
\begin{align}
J_1  & \geq  n\log (\pi e)    \label{eq:Jb1} \\
J_3  & \geq  n\log (\pi e)   \label{eq:Jb3}  \\
J_2  & \leq  n\log (4\pi e)    \label{eq:Jb2}  \\
J_4  & \leq  n\log (\pi e(1+ P^{\alpha_{22}-\alpha_{12}} + P^{\alpha_{21} -(\alpha_{11}- \alpha_{12})^+})).    \label{eq:Jb4}  
\end{align}
\end{lemma}
\begin{proof}
See Appendix~\ref{app:Jbounds}.
\end{proof}

\section{Conclusion and discussion}   \label{sec:conclusion}

This work showed that a simple scheme without cooperative jamming (i.e., GWC-TIN scheme) can indeed achieve  the secure sum capacity to within a constant gap in a  regime  (see~\eqref{eq:capGaussian1} and \eqref{eq:capGaussian2}) of the two-user Gaussian interference channel. 
In this GWC-TIN scheme,  each transmitter uses a Gaussian wiretap codebook, while each receiver treats interference as noise when decoding the desired message. 
In this simple scheme,  the transmitters do not need to know  the information of the \emph{channel phases}. 
For the deterministic interference channel model, this work identified a regime (see~\eqref{eq:capdetcond1} and \eqref{eq:capdetcond2}) in which  a simple scheme without using cooperative jamming  (i.e., WoCJ scheme) is optimal in terms of secure sum capacity.  For the symmetric case of the deterministic model, the identified regime (i.e., $0\leq \alpha \leq 2/3$) is indeed a sufficient and necessary regime for WoCJ scheme to be  optimal. 
In the future work, it would be interesting to determine whether the conditions in \eqref{eq:capdetcond1} and \eqref{eq:capdetcond2} are  also necessary for WoCJ scheme to be optimal, for the general deterministic interference channel.  For the general Gaussian interference channel, again, it would be interesting to determine whether the conditions \eqref{eq:capGaussian1} and \eqref{eq:capGaussian2} are  also necessary for GWC-TIN scheme to be optimal in terms of secure sum capacity to within a constant gap.
Another direction of the future  work is to study on the optimality of  the secure communication without using cooperative jamming in some other channels, such as $k$-user interference channel and wiretap channel with helper.

\appendices

\section{Proof of Lemma~\ref{lm:Jbounds}} \label{app:Jbounds}

Remind that  $J_1= \hen(s_{11}^{\bln}| y^{\bln}_2, w_1)$, $ J_2 = \hen(s_{11}^{\bln} | y^{\bln}_1, w_1) $, $J_3 = \hen(s_{12}^{\bln} | y^{\bln}_2, s_{11}^{\bln}, w_1)$, $J_4= \hen( y^{\bln}_2| s_{12}^{\bln}, s_{11}^{\bln}, w_1)$,  $s_{11}(t)  =  \sqrt{P^{(\alpha_{11}-\alpha_{12})^+}} e^{j\theta_{11}} x_{1}(t) +  \tilde{z}_{1}(t)$, and $ s_{12}(t) =  \sqrt{P^{\alpha_{12}}} e^{j\theta_{12}} x_{2}(t)  +z_{1}(t)$.  
At first we focus on the lower bound of $J_1$:
\begin{align}
J_1  & = \hen(s_{11}^{\bln}| y^{\bln}_2, w_1)  \nonumber \\
       & \geq \hen(s_{11}^{\bln}|  x_1^{\bln}, y^{\bln}_2, w_1)     \label{eq:Jb4256}  \\
      & = \hen(   \{   \sqrt{P^{(\alpha_{11}-\alpha_{12})^+}} e^{j\theta_{11}} x_{1}(t) +  \tilde{z}_{1}(t) \}_{t=1}^{\bln}  |  x_1^{\bln}, y^{\bln}_2, w_1)     \nonumber  \\
       & = \hen(   \{ \tilde{z}_{1}(t) \}_{t=1}^{\bln}  |  x_1^{\bln}, y^{\bln}_2, w_1)     \label{eq:Jb8488}  \\
       & = \hen(   \{ \tilde{z}_{1}(t) \}_{t=1}^{\bln} )    \nonumber  \\
       & =  n\log (\pi e) \nonumber 
\end{align}
where \eqref{eq:Jb4256} follows from the fact that conditioning reduces differential entropy;   \eqref{eq:Jb8488} follows from the fact that $\hen(a|b)=\hen(a-b |b)$ for any continuous random variables $a$ and $b$; the last equality holds true because $\hen( \tilde{z}_{1}(t)) = \log (\pi e) $.   Similarly, we have 
\begin{align}
J_3  & = \hen(s_{12}^{\bln} | y^{\bln}_2, s_{11}^{\bln}, w_1)  \nonumber \\
       & \geq  \hen(s_{12}^{\bln} |  x_2^{\bln}, y^{\bln}_2, s_{11}^{\bln}, w_1)    \nonumber \\
       & = \hen(   \{ z_{1}(t) \}_{t=1}^{\bln} )    \nonumber  \\
       & =  n\log (\pi e). \nonumber 
\end{align}

Now we focus on the  upper bound of $J_2$:
\begin{align}
J_2  & = \hen(s_{11}^{\bln} | y^{\bln}_1, w_1)  \nonumber \\
        & = \sum_{t=1}^{\bln}  \hen(s_{11}(t) | s_{11}^{t-1}, y^{\bln}_1, w_1)  \label{eq:Jb0263}  \\
        & \leq  \sum_{t=1}^{\bln}  \hen(s_{11}(t) |  y_1(t))  \label{eq:Jb9255}  \\
        & =  \sum_{t=1}^{\bln}  \hen\bigl(s_{11}(t) -  \sqrt{P^{-\alpha_{12} }}y_1(t) |  y_1(t)\bigr)  \label{eq:Jb2595}  \\
        & =  \sum_{t=1}^{\bln}  \hen \Bigl(  \tilde{z}_{1}(t)  + (  \sqrt{P^{(\alpha_{11}-\alpha_{12})^+}} - \sqrt{P^{\alpha_{11}-\alpha_{12} }}  )  e^{j\theta_{11}} x_{1}(t)  -    e^{j\theta_{12}} x_{2}(t)  - \sqrt{P^{-\alpha_{12} }} z_{1}(t)    \big|  y_1(t) \Bigr)  \nonumber  \\
        & \leq  \sum_{t=1}^{\bln}  \hen \Bigl(    \tilde{z}_{1}(t)  + (  \sqrt{P^{(\alpha_{11}-\alpha_{12})^+}} - \sqrt{P^{\alpha_{11}-\alpha_{12} }}  )  e^{j\theta_{11}} x_{1}(t)  -    e^{j\theta_{12}} x_{2}(t)  - \sqrt{P^{-\alpha_{12} }} z_{1}(t)     \Bigr)   \label{eq:Jb5367}  \\
       & \leq   n \log (\pi e(2+  \underbrace{(  \sqrt{P^{(\alpha_{11}-\alpha_{12})^+}} - \sqrt{P^{\alpha_{11}-\alpha_{12} }}  )^2 }_{\leq 1}+    \underbrace{P^{-\alpha_{12} } }_{\leq 1} ))  \label{eq:Jb44112} \\
        & \leq   n \log (4\pi e)    \label{eq:Jb00998} 
\end{align}
where \eqref{eq:Jb0263} results from chain rule;
 \eqref{eq:Jb9255}  and \eqref{eq:Jb5367} follow from the fact that conditioning reduces differential entropy;
  \eqref{eq:Jb2595} uses the fact that $\hen(a|b)=\hen(a-\beta b |b)$ for a constant $\beta$; 
  \eqref{eq:Jb44112}  follows from  the fact that  $ \hen \bigl(   \tilde{z}_{1}(t)  +  \beta_0  e^{j\theta_{11}} x_{1}(t)  - \beta_1  e^{j\theta_{12}} x_{2}(t)  -  \beta_2 z_{1}(t)   \bigr)  \leq  \log (\pi e(1+ \beta_0^2+ \beta_1 ^2+   \beta_2^2  ))$ for constants $\beta_0$, $\beta_1$ and $\beta_2$;
  \eqref{eq:Jb00998}  uses the identities that $0\leq \sqrt{P^{(\alpha_{11}-\alpha_{12})^+}} - \sqrt{P^{\alpha_{11}-\alpha_{12} }}   \leq  1$ and that  $ P^{-\alpha_{12} } \leq 1 $. Remind that  $P\geq 1$, $\alpha_{k\ell} \geq 0, \forall k, \ell  \in \{1,2\}$, and $(\bullet)^+= \max\{0, \bullet\}$.

Similarly, we have the following upper bound on $J_4$:
\begin{align}
J_4  & = \hen( y^{\bln}_2| s_{12}^{\bln}, s_{11}^{\bln}, w_1)  \nonumber \\
        & = \sum_{t=1}^{\bln}  \hen(y_{2}(t) | y_{2}^{t-1}, s_{12}^{\bln}, s_{11}^{\bln}, w_1)  \label{eq:Jb02631}  \\
        & \leq  \sum_{t=1}^{\bln}  \hen(y_{2}(t) | s_{12}(t), s_{11}(t) )  \label{eq:Jb92551}  \\
        & =  \sum_{t=1}^{\bln}  \hen\bigl(y_{2}(t) -  \sqrt{P^{\alpha_{22}-\alpha_{12}}}  e^{j(\theta_{22}- \theta_{12})} s_{12}(t)  -  \sqrt{P^{\alpha_{21} -(\alpha_{11} - \alpha_{12} )^+}}  e^{j(\theta_{21}- \theta_{11})} s_{11}(t)   | s_{12}(t), s_{11}(t) \bigr)  \label{eq:Jb25951}  \\
        & =  \sum_{t=1}^{\bln}  \hen \Bigl( z_{2}(t)   -  \sqrt{P^{\alpha_{22}-\alpha_{12}}}  e^{j(\theta_{22}- \theta_{12})}  z_{1}(t)  -  \sqrt{P^{\alpha_{21} -(\alpha_{11} - \alpha_{12} )^+}}  e^{j(\theta_{21}- \theta_{11})}  \tilde{z}_{1}(t)       \big| s_{12}(t), s_{11}(t)  \Bigr)  \nonumber  \\
        & \leq  \sum_{t=1}^{\bln}   \hen \Bigl( z_{2}(t)   -  \sqrt{P^{\alpha_{22}-\alpha_{12}}}  e^{j(\theta_{22}- \theta_{12})}  z_{1}(t)  -  \sqrt{P^{\alpha_{21} -(\alpha_{11} - \alpha_{12} )^+}}  e^{j(\theta_{21}- \theta_{11})}  \tilde{z}_{1}(t)   \Bigr)   \label{eq:Jb53671}  \\
       & \leq    n\log (\pi e(1+ P^{\alpha_{22}-\alpha_{12}}  +  P^{\alpha_{21} -(\alpha_{11} - \alpha_{12} )^+} ))  \nonumber 
\end{align}
where \eqref{eq:Jb02631} results from chain rule;
 \eqref{eq:Jb92551}  and \eqref{eq:Jb53671} follow from the fact that conditioning reduces differential entropy;
     \eqref{eq:Jb25951} uses the fact that $\hen(a|b,c)=\hen(a-\beta_1b -\beta_2 c |b, c)$ for  constants $\beta_1$ and  $\beta_2$;
 the last inequality  stems from the fact that $\hen \bigl( z_{2}(t)   -  \beta_3  e^{j(\theta_{22}- \theta_{12})} z_{1}(t)    -  \beta_4  e^{j(\theta_{21}- \theta_{11})} \tilde{z}_{1}(t)     \bigr)   \leq  \log (\pi e(1+ \beta_3^2 + \beta_4^2))$ for  constants $\beta_3$ and  $\beta_4$. At this point we complete the proof.

 \section{Proof of \eqref{eq:detup1}, \eqref{eq:detup2}, \eqref{eq:detup4} and \eqref{eq:detup5} in Lemma~\ref{lm:detconverse}   \label{sec:conversesome}}

\subsection{Proof of bounds \eqref{eq:detup1} and \eqref{eq:detup2} \label{sec:detup12} }

Let us at first prove bound \eqref{eq:detup1}. Note that  bound \eqref{eq:detup2} can be proved in a similar way. 
Beginning with Fano's inequality, we can bound the rate of receiver~1 as:
\begin{align}
  \bln R_1   
  &\leq  \Imu(w_1; \yv^{\bln}_1) + \bln \epsilon_{1,n}  \non \\
  &\leq  \Imu(w_1; \xv^{\bln}_{1,1: m_{11}}) + \bln \epsilon_{1,n}    \label{eq:2remove}  \\
    &\leq  \Imu(w_1; \xv^{\bln}_{1,1: m_{11}}, \xv^{\bln}_{2,1: (\min\{m_{11}, m_{21}\}- ( m_{21} - m_{22} )^+)^+}, \yv^{\bln}_2) + \bln \epsilon_{1,n}    \label{eq:2genie}  \\
    &\leq   \Imu(w_1; \xv^{\bln}_{1,1: m_{11}}, \xv^{\bln}_{2,1: (\min\{m_{11}, m_{21}\}- ( m_{21} - m_{22} )^+)^+}, \yv^{\bln}_2) - \Imu(w_1; \yv^{\bln}_2)  +  \bln \epsilon  + \bln \epsilon_{1,n}  \label{eq:2R1sec88}    \\
&=  \Imu(w_1; \xv^{\bln}_{1,1: m_{11}}, \xv^{\bln}_{2,1: (\min\{m_{11}, m_{21}\}- ( m_{21} - m_{22} )^+)^+} |  \yv^{\bln}_2)    +  \bln \epsilon  + \bln \epsilon_{1,n}   \non \\ 
&=  \underbrace{ \Imu(w_1;  \xv^{\bln}_{2,1: (\min\{m_{11}, m_{21}\}- ( m_{21} - m_{22} )^+)^+} |  \yv^{\bln}_2)}_{\leq \bln \cdot  (\min\{m_{11}, m_{21}\}- ( m_{21} - m_{22} )^+)^+}  + \underbrace{ \Imu(w_1;   \xv^{\bln}_{1,1: m_{11}} |  \yv^{\bln}_2,   \xv^{\bln}_{2,1: (\min\{m_{11}, m_{21}\}- ( m_{21} - m_{22} )^+)^+} )}_{
 \leq   \bln \cdot  (   m_{11}  - \min\{m_{11}, m_{21}\}  ) }    \non\\&\quad +  \bln \epsilon  + \bln \epsilon_{1,n}   \non\\
& \leq \bln \cdot(\min\{m_{11}, m_{21}\}- ( m_{21} - m_{22} )^+)^+  +  \bln \cdot  (   m_{11}  - \min\{m_{11}, m_{21}\}) + \bln \epsilon  + \bln \epsilon_{1,n}   \label{eq:2R1s224}   \\
& =   \bln \cdot  (   m_{11}  - \min\{m_{11}, m_{21}, ( m_{21} - m_{22} )^+ \} ) + \bln \epsilon  + \bln \epsilon_{1,n}    \label{eq:2R1s73665}   \\
& =   \bln \cdot  \max\{ 0, \   m_{11}  -  ( m_{21} - m_{22} )^+  \} + \bln \epsilon  + \bln \epsilon_{1,n}   \non
\end{align}
where   $\lim_{n\to\infty} \epsilon_{1,n} = 0$;
\eqref{eq:2remove} stems from the Markov chain  of $w_1 \to \xv^{\bln}_{1,1: m_{11}} \to  \yv^{\bln}_1$;  
\eqref{eq:2genie} results from the fact that adding information does not decrease the mutual information;
 \eqref{eq:2R1sec88}  results  from  the secrecy constraint, i.e., $ \Imu(w_1; \yv^{\bln}_2)  \leq \bln \epsilon$ for an  arbitrary small  $\epsilon$ (cf.~\eqref{eq:defsecrecy1});
\eqref{eq:2R1s224}  follows from  the fact that $\Imu(w_1;   \xv^{\bln}_{2,1: (\min\{m_{11}, m_{21}\}- ( m_{21} - m_{22} )^+)^+} |  \yv^{\bln}_2) $  $  \leq \bln \cdot (\min\{m_{11}, m_{21}\}- ( m_{21} - m_{22} )^+)^+$,  and  the fact that $\xv^{\bln}_{1,1: \min\{m_{11}, m_{21}\}}$ can be reconstructed from $\{ \yv^{\bln}_2,   \xv^{\bln}_{2,1: (\min\{m_{11}, m_{21}\}- ( m_{21} - m_{22} )^+)^+} \}$  (see Fig.~\ref{fig:ICdety2a32});
\eqref{eq:2R1s73665} follows from the identity that $(a_1 - a_2)^+ +a_3 - a_1  = a_3 -\min\{a_1, a_2\}$ for any real numbers $a_1, a_2$ and $a_3$.
Letting  $n\to \infty, \  \epsilon_{1,n} \to 0 $ and  $\epsilon \to 0$, it gives bound \eqref{eq:detup1}.  By interchanging the roles of user~1 and user~2, bound \eqref{eq:detup2} can be proved in a similar way.

\begin{figure}[t!]
\centering
\includegraphics[width=11cm]{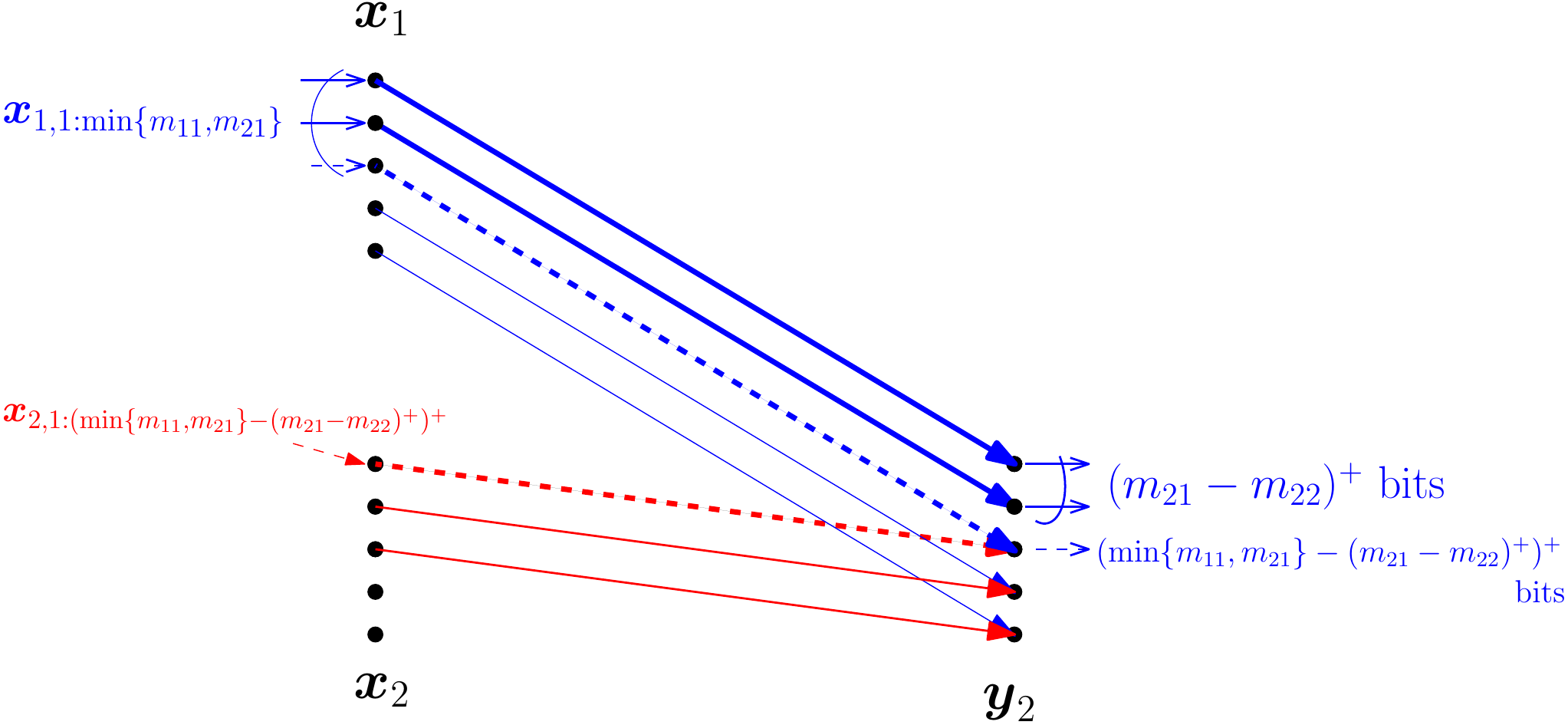}
\caption{Deterministic channel model for receiver~2.   $\xv_{1,1: \min\{m_{11}, m_{21}\}}$ can be reconstructed from $\{\yv_2,   \xv_{2,1: (\min\{m_{11}, m_{21}\}- ( m_{21} - m_{22} )^+)^+} \}$. For example, when $m_{11}=m_{22} = 3$ and $m_{12}=m_{21} = 5$, then $\xv_{1,1: 3}$ can be reconstructed from $\{\yv_2,   \xv_{2,1: 1} \}$. } 
\label{fig:ICdety2a32}
\end{figure}

\subsection{Proof of bounds \eqref{eq:detup4} and \eqref{eq:detup5} \label{sec:detup45} }

In what follows we will prove bound \eqref{eq:detup4}. Bound \eqref{eq:detup5} can be proved in a similar way. 
Beginning with Fano's inequality, we can bound the rate of receiver~1 as:
\begin{align}
 & \bln R_1 -  \bln \epsilon_{1,n}    \non\\
&\leq  \Imu(w_1; \yv^{\bln}_1)   \non \\
  &\leq  \Imu(w_1; \yv^{\bln}_1)    - \Imu(w_1; \yv^{\bln}_2) + \bln \epsilon    \label{eq:2R1sec325}  \\
  &\leq  \Imu(w_1; \yv^{\bln}_1, \xv^{\bln}_{1, 1: \max\{m_{11}, m_{21}\}}, \xv^{\bln}_{2, 1: \max\{m_{12}, m_{22}\}},    \yv^{\bln}_2 )   - \Imu(w_1; \yv^{\bln}_2) + \bln \epsilon    \label{eq:2R1add223}  \\
  &=   \Imu(w_1; \yv^{\bln}_1, \xv^{\bln}_{1, 1: \max\{m_{11}, m_{21}\}}, \xv^{\bln}_{2, 1: \max\{m_{12}, m_{22}\}} |   \yv^{\bln}_2 )    + \bln \epsilon   \non  \\
  &=    \Imu(w_1; \xv^{\bln}_{1, 1: \max\{m_{11}, m_{21}\}}, \xv^{\bln}_{2, 1: \max\{m_{12}, m_{22}\}} |   \yv^{\bln}_2 ) +\underbrace{ \Imu(w_1; \yv^{\bln}_1 | \xv^{\bln}_{1, 1: \max\{m_{11}, m_{21}\}}, \xv^{\bln}_{2, 1: \max\{m_{12}, m_{22}\}}, \yv^{\bln}_2 ) }_{=0}+ \bln \epsilon   \non  \\
  &=   \Imu(w_1; \xv^{\bln}_{1, 1: \max\{m_{11}, m_{21}\}}, \xv^{\bln}_{2, 1: \max\{m_{12}, m_{22}\}} |   \yv^{\bln}_2 ) + \bln \epsilon    \label{eq:2R1Markove218}   \\
    &\leq   \Hen(\xv^{\bln}_{1, 1: \max\{m_{11}, m_{21}\}}, \xv^{\bln}_{2, 1: \max\{m_{12}, m_{22}\}} |   \yv^{\bln}_2 )  + \bln \epsilon    \label{eq:2R1Markove3322}   \\
  &=    \Hen(\xv^{\bln}_{1, 1: \max\{m_{11}, m_{21}\}}, \xv^{\bln}_{2, 1: \max\{m_{12}, m_{22}\}},  \yv^{\bln}_2 )  -  \Hen( \yv^{\bln}_2 )   + \bln \epsilon    \non   \\
  &=    \Hen(\xv^{\bln}_{1, 1: \max\{m_{11}, m_{21}\}}, \xv^{\bln}_{2, 1: \max\{m_{12}, m_{22}\}})  -  \Hen( \yv^{\bln}_2 )   + \bln \epsilon    \label{eq:2R1Markove763}   
 \end{align}
where \eqref{eq:2R1sec325}  results  from a secrecy constraint (cf.~\eqref{eq:defsecrecy1});
\eqref{eq:2R1add223}  stems from the fact that adding information does not decrease the mutual information;  
\eqref{eq:2R1Markove3322} follows from the nonnegativity of entropy;
\eqref{eq:2R1Markove218} and \eqref{eq:2R1Markove763} follow  from the Markov chain  of $\{w_1, w_2\} \to \{ \xv^{\bln}_{1, 1: \max\{m_{11}, m_{21}\}} , \xv^{\bln}_{2, 1: \max\{m_{12}, m_{22}\}}\}   \to  \{\yv^{\bln}_1, \yv^{\bln}_2\}$.  On the other hand, we have 
\begin{align}
   \bln R_1    
&\leq  \Imu(w_1; \yv^{\bln}_1)  + \bln \epsilon_{1,n}  \label{eq:Fano2335} \\
  &\leq  \Imu(\xv^{\bln}_{1}; \yv^{\bln}_1 )  + \bln \epsilon_{1,n}   \label{eq:2R1Markove82435}  \\
  &=    \Hen(\yv^{\bln}_1 )  -  \Hen( \yv^{\bln}_1 | \xv^{\bln}_{1})   +  \bln \epsilon_{1,n}    \non   \\
  &=    \Hen(\yv^{\bln}_1 )  -  \Hen( \xv^{\bln}_{2, 1: m_{12}} | \xv^{\bln}_{1})   +  \bln \epsilon_{1,n}     \label{eq:2R1y122}  \\
  &=    \Hen(\yv^{\bln}_1 )  -  \Hen( \xv^{\bln}_{2, 1:m_{12}} )   + \bln \epsilon_{1,n}     \label{eq:2R1ind254}  
 \end{align}
where \eqref{eq:Fano2335} results from  Fano's inequality;
\eqref{eq:2R1Markove82435}  follows from  the Markov chain  of $w_1 \to  \xv^{\bln}_{1}  \to  \yv^{\bln}_1$;  
\eqref{eq:2R1y122} results from the definition of $\yv_1(t) =   S^{q- m_{11}} \xv_1(t)  \oplus   S^{q- m_{12}} \xv_2(t)$ (cf.~\eqref{eq:detTHIC1});
\eqref{eq:2R1ind254}  follows from the independence between $\xv^{\bln}_{1}$ and $\xv^{\bln}_{2}$.
In a similar way, we have 
\begin{align}
   \bln R_2    
&\leq    \Hen(\yv^{\bln}_2 )  -  \Hen( \xv^{\bln}_{1, 1:m_{21}} )   + \bln \epsilon_{2,n} .    \label{eq:2R2ind888}  
 \end{align}
Finally, by combing \eqref{eq:2R1Markove763}, \eqref{eq:2R1ind254} and \eqref{eq:2R2ind888}, it gives
\begin{align}
 &  2\bln R_1 + \bln R_2 -  2\bln \epsilon_{1,n} - \bln \epsilon_{2,n} - \bln \epsilon    \non  \\
  &\leq     \underbrace{\Hen(\yv^{\bln}_1 )}_{\leq  \bln \cdot \max\{ m_{11}, m_{12} \}}  +  \underbrace{  \Hen(\xv^{\bln}_{1, 1: \max\{m_{11}, m_{21}\}})    -  \Hen( \xv^{\bln}_{1, 1: m_{21}} )}_{ \leq    \bln \cdot (  m_{11}  - m_{21})^+   }       +  \underbrace{  \Hen(\xv^{\bln}_{2, 1: \max\{m_{12}, m_{22}\}})    -  \Hen( \xv^{\bln}_{2, 1: m_{12}} )   }_{ \leq    \bln \cdot (  m_{22}  - m_{12})^+ }     \non \\
  &\leq       \bln \cdot \max\{ m_{11}, m_{12} \}  +  \bln \cdot (  m_{11}  - m_{21})^+ +    \bln \cdot (  m_{22}  - m_{12})^+   \label{eq:2R1sum333}
 \end{align}
where \eqref{eq:2R1sum333} follows from the facts that $\Hen(\yv^{\bln}_1 ) \leq  \bln \cdot \max\{ m_{11}, m_{12} \}$,  that  $\Hen(\xv^{\bln}_{1, 1: \max\{m_{11}, m_{21}\}})    -  \Hen( \xv^{\bln}_{1, 1: m_{21}} )  \leq \bln \cdot (\max\{m_{11}, m_{21}\} - m_{21}) = \bln \cdot (  m_{11}  - m_{21})^+ $, and that $\Hen(\xv^{\bln}_{2, 1: \max\{m_{12}, m_{22}\}})    -  \Hen(  \xv^{\bln}_{2, 1: m_{12}})  \leq \bln \cdot (  m_{22}  - m_{12})^+ $.
By setting  $n\to \infty, \  \epsilon_{1,n}, \epsilon_{2,n} \to 0 $ and  $\epsilon \to 0$, it gives bound \eqref{eq:detup4}.
By interchanging the roles of user~1 and user~2, bound \eqref{eq:detup5} can be proved in a similar way.

\section{A sketch of the cooperative jamming schemes for Theorem~\ref{thm:capacitydet}  \label{sec:achisome}}

For the two-user \emph{symmetric}  deterministic interference channel, the work in \cite{GTJ:15} has proposed some  schemes \emph{with cooperative jamming} that are optimal when $ \frac{2}{3} < \alpha <1$ and  $ 1 < \alpha <2$, in terms of secure sum secure capacity.
This section just provides  a sketch of these cooperative jamming schemes, as more details could be found in  \cite{GTJ:15}.
In these cooperative jamming schemes, each transmitter generally sends one or more of the following signals: 1) private data signal, which can only be seen by its desired receiver; 2) common data signal, which can be received by both receivers; 3) jamming signal, which can be used to jam the unintended common data signal to guarantee secrecy.

For the case of $ \frac{2}{3} < \alpha \leq \frac{3}{4}$,  at each time each transmitter  sends a total of $2\mc- \md$ bits of data, consisting of $\md- \mc$ bits of private data and $3\mc- 2\md$ bits of common data. 
In this case, the transmission of both  private data and common data is secure from the corresponding  eavesdropper by using cooperative jamming (see Fig.~\ref{fig:schemeICdet34} on the setting with $\md=4$ and $\mc=3$).

\begin{figure}[t!]
\centering
\includegraphics[width=7cm]{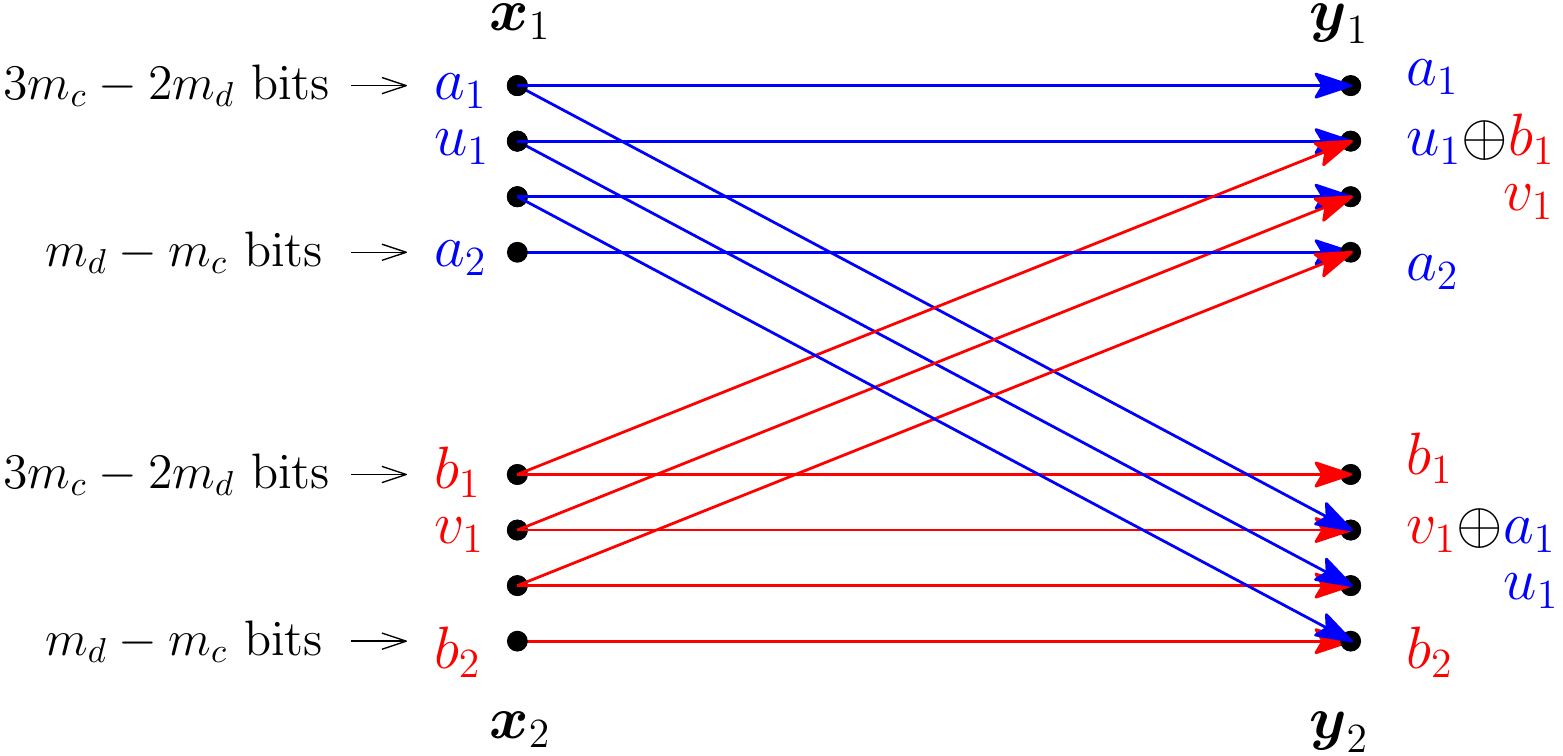}
\caption{Optimal scheme for the symmetric deterministic channel with $\md = 4$ and $\mc =3$.  At each time, each transmitter  sends  $\md- \mc=1$ bit of private data and $3\mc- 2\md = 1 $ bit of common data.   The transmission of  private data  ($a_2$ for transmitter~1 and $b_2$ for transmitter~2) is secure from the unintended receiver.  The transmission of common data ($a_1$  for transmitter~1 and $b_1$ for transmitter~2) is also secure by using cooperative jamming signal ($u_1$ for transmitter~1 and $v_1$ for transmitter~2).} 
\label{fig:schemeICdet34}
\end{figure}

\begin{figure}[t!]
\centering
\includegraphics[width=7cm]{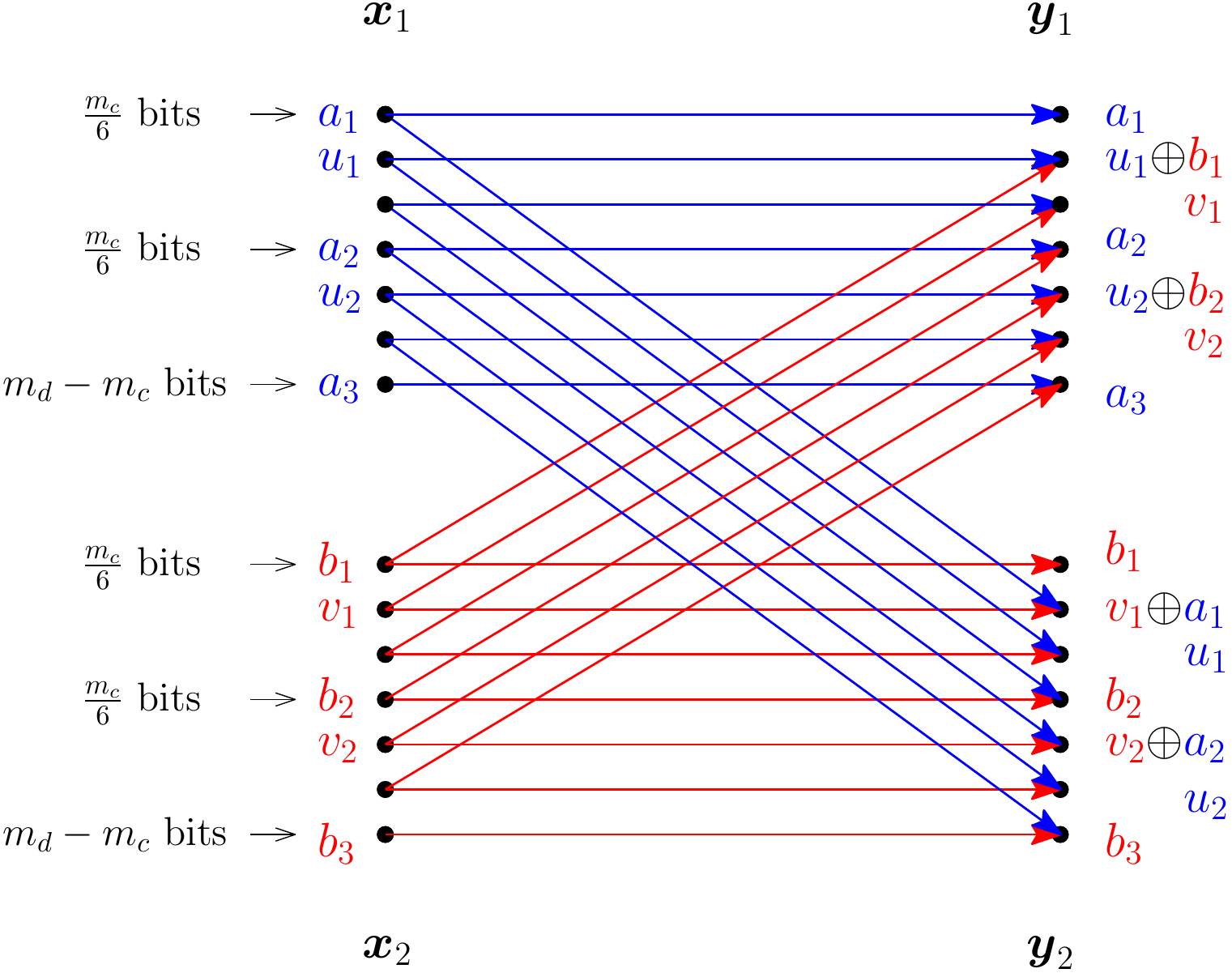}
\caption{Optimal scheme for the symmetric deterministic channel with $\md = 7$ and $\mc =6$.  At each time,  each transmitter  sends a total of $\md- \frac{2}{3}\mc= 3$ bits of data, consisting of $\md- \mc=1$ bit of private data and $\frac{\mc}{3}= 2$ bits of common data.   The transmission of private data  ($a_3$ for transmitter~1 and $b_3$ for transmitter~2) is secure from the unintended receiver.  The transmission of common data ($a_1, a_2$  for transmitter~1 and $b_1, b_2$ for transmitter~2) is also secure by using cooperative jamming signal ($u_1, u_2$ for transmitter~1 and $v_1, v_2$ for transmitter~2).} 
\label{fig:schemeICdet67}
\end{figure}

\begin{figure}[t!]
\centering
\includegraphics[width=7cm]{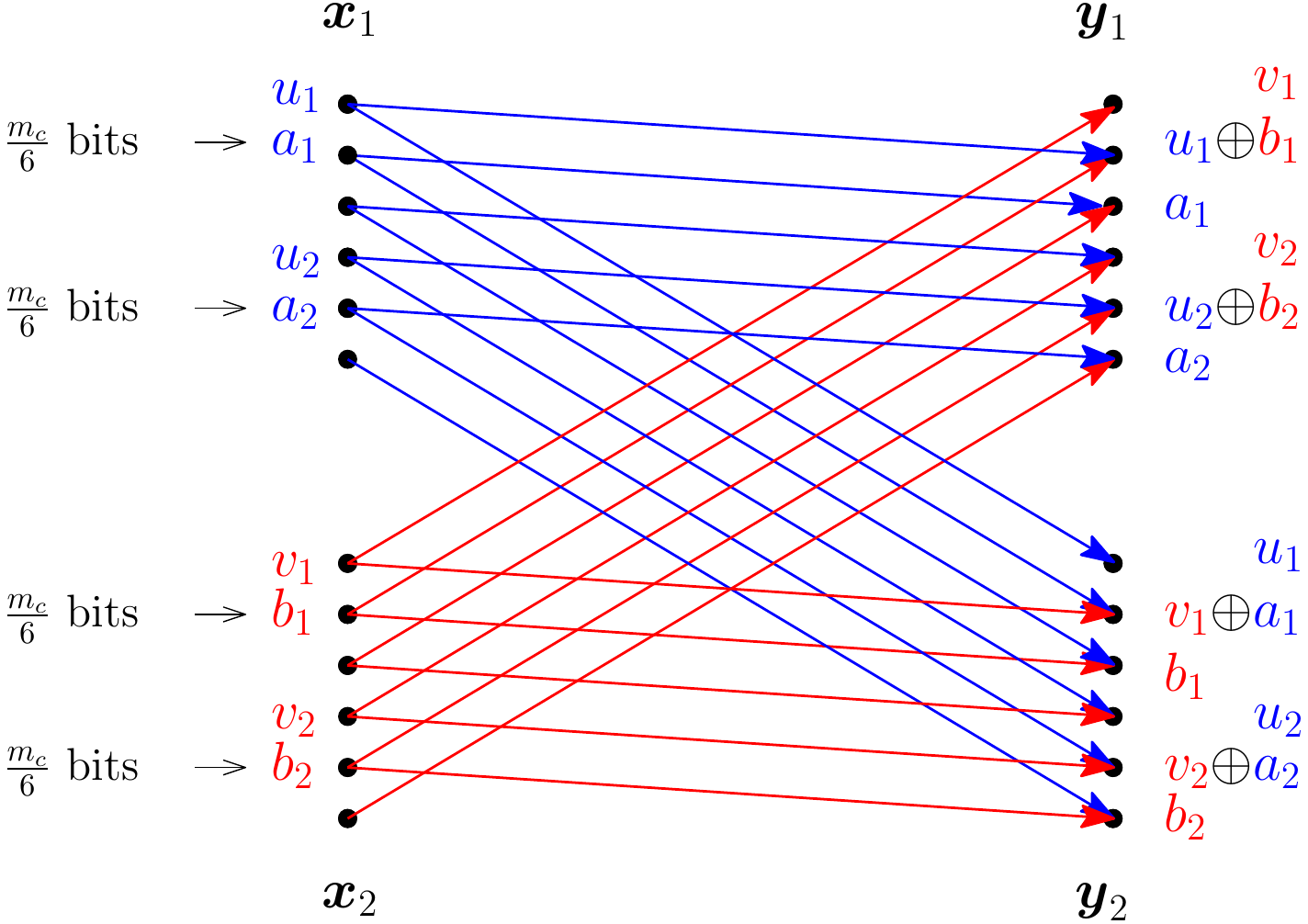}
\caption{Optimal scheme for the symmetric deterministic channel with $\md = 5$ and $\mc =6$.  At each time, each transmitter  sends a total of $ \frac{\mc}{3}$ bits of common data.  The transmission of common data ($a_1, a_2$  for transmitter~1 and $b_1, b_2$ for transmitter~2) is secure by utilizing cooperative jamming.} 
\label{fig:schemeICdet65}
\end{figure}

\begin{figure}[t!]
\centering
\includegraphics[width=8cm]{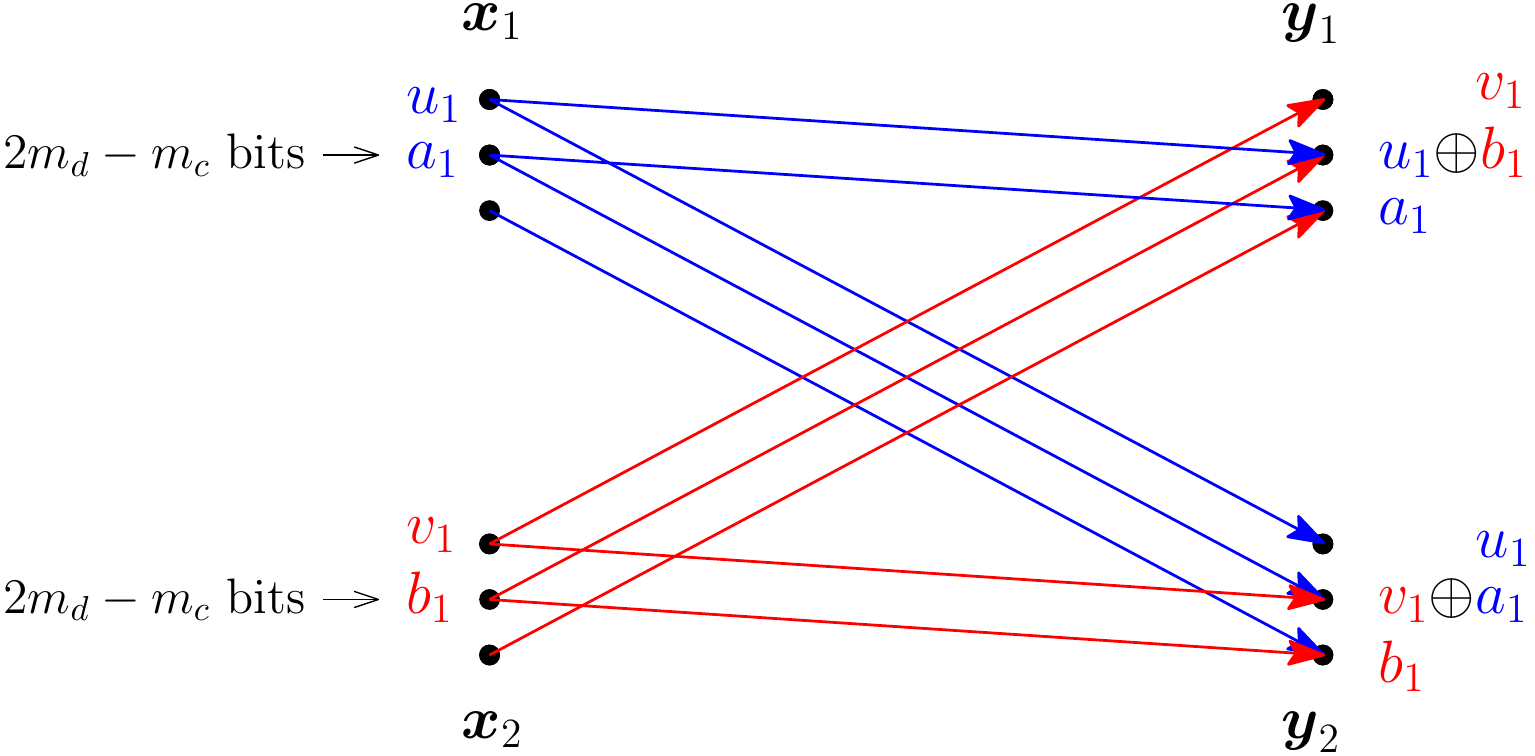}
\caption{Optimal scheme for the symmetric deterministic channel with $\md = 2$ and $\mc =3$.  At each time, each transmitter  sends a total of  $2\md- \mc =1$ bit of common data that is secure from the corresponding  eavesdropper by utilizing cooperative jamming.} 
\label{fig:schemeICdet32}
\end{figure}

For the case of $ \frac{3}{4} < \alpha \leq  1 $,  at each time each transmitter  sends a total of $\md- \frac{2}{3}\mc$ bits of data, consisting of $\md- \mc$ bits of private data and $\frac{\mc}{3}$ bits of common data. For simplicity of exposition, we focus on the setting with  $\md=7$ and $\mc=6$, and depict the optimal scheme in Fig.~\ref{fig:schemeICdet67}. As can be seen in Fig.~\ref{fig:schemeICdet67}, the transmission of both  private data and common data is secure from the corresponding  eavesdropper.

For the case of $ 1 < \alpha <  \frac{3}{2} $,  each transmitter  sends a total of $ \frac{\mc}{3}$ bits of common data per channel time.  The  optimal scheme  is depicted in Fig.~\ref{fig:schemeICdet65} on the setting with  $\md=5$ and $\mc=6$. 

For the case of $ \frac{3}{2} \leq  \alpha <   2$,  each transmitter  sends a total of  $2\md- \mc$ bits of common data per channel time.   The  optimal scheme  is depicted in  Fig.~\ref{fig:schemeICdet32} on the setting with  $\md=2$ and $\mc=3$.

\section*{Acknowledgement}
We wish to thank Ayfer \"Ozg\"ur and Andrea Goldsmith for helpful comments during the early stage of this work. We also wish to thank Chunhua Geng, Syed Ali Jafar and Zhiying Wang for helpful comments on the early version of this work.



\end{document}